\documentclass[aps,pra,10pt,letterpaper,twocolumn,tightenlines,superscriptaddress,notitlepage,longbibliography]{revtex4-2}

\usepackage[utf8]{inputenc}
\usepackage[english]{babel}
\usepackage[T1]{fontenc}
\usepackage{amsmath}
\makeatletter
\newcommand{\pushright}[1]{\ifmeasuring@#1\else\omit\hfill$\displaystyle#1$\fi\ignorespaces}
\newcommand{\pushleft}[1]{\ifmeasuring@#1\else\omit$\displaystyle#1$\hfill\fi\ignorespaces}
\makeatother

\makeatletter 
    
\renewcommand\onecolumngrid{% <<<<<<
\do@columngrid{one}{\@ne}%
\def\set@footnotewidth{\onecolumngrid}% <<<<<<<<<<<<<<<<
\def\footnoterule{\kern-6pt\hrule width 1.5in\kern6pt}%
}

\renewcommand\twocolumngrid{% <<<<<<
        \def\footnoterule{% restore rule
        \dimen@\skip\footins\divide\dimen@\thr@@
        \kern-\dimen@\hrule width.5in\kern\dimen@}
        \do@columngrid{mlt}{\tw@}
}%

\makeatother    
\allowdisplaybreaks

\usepackage{tikz}
\usetikzlibrary{positioning,calc,arrows.meta}
\usepackage{lipsum}
\usepackage{amssymb}
\usepackage{color}
\usepackage{mathrsfs}
\usepackage{amsbsy}
\usepackage{amsthm}
\usepackage{array}
\newcolumntype{M}[1]{>{\centering\arraybackslash}m{#1}}
\newcolumntype{N}{@{}m{0pt}@{}}
\usepackage{balance}
\usepackage{thmtools,thm-restate}
\usepackage{graphicx}
\usepackage{tikz}

\usepackage{comment}

\usepackage{bbm}
\usepackage{bm}
\usepackage{epsfig}
\usepackage{xfrac}
\usepackage{xcolor}
\usepackage{enumerate}
\usepackage[shortlabels]{enumitem}
\usepackage{graphicx}% Include figure files
\usepackage{bm}% bold math
\usepackage{hyperref}
\usepackage{booktabs}
\usepackage{makecell}

\usepackage{subfigure}
\hypersetup{
colorlinks=true,
linkcolor=blue,
filecolor=blue,
citecolor=blue,  
urlcolor=blue,
}

\newtheorem{theorem}{Theorem}
\newtheorem{theorem-main}{Theorem}

\newtheorem{lemma}{Lemma}

\newcommand{\cU}{\check{U}}

\newcommand{\htheta}{\hat{\theta}}

\DeclareFontFamily{U}{mathx}{}
\DeclareFontShape{U}{mathx}{m}{n}{<-> mathx10}{}
\DeclareSymbolFont{mathx}{U}{mathx}{m}{n}
\DeclareMathAccent{\widecheck}{0}{mathx}{"71}

\newcommand{\cpsi}{\widecheck{\psi}}
\newcommand{\ctheta}{{\check{\theta}}}
\newcommand{\cphi}{{\check{\phi}}}
\newcommand{\zphi}{{\phi^0}}
\newcommand{\ztheta}{{\theta^0}}
\newcommand{\zU}{{U^0}}
\newcommand{\dpu}{{d}_{\mathrm{PU}}}
\newcommand{\zvarphi}{{\varphi^0}}
\newcommand{\cvarphi}{\check{\varphi}}

\DeclareMathOperator{\Tr}{Tr}

\usepackage{stackengine}

\usepackage{braket}
\renewcommand{\v}[1]{\ensuremath{\mathbf{#1}}} % for vectors
\newcommand{\abs}[1]{\left| #1 \right|}
\newcommand{\norm}[1]{\left\| #1 \right\|} % for norm
\newcommand{\trace}{\mathrm{Tr}}
\renewcommand{\bell}{{\bar\ell}}

\newcommand{\mX}{{\mathcal{X}}}
\newcommand{\mH}{{\mathcal{H}}}

\newcommand{\mU}{{\mathcal{U}}}

\newcommand{\mM}{{\mathcal{M}}}

\newcommand{\tW}{{\widetilde{W}}}

\newcommand{\mE}{{\mathcal{E}}}

\newcommand{\id}{{\mathbbm{1}}}

\newcommand{\vX}{{\v{X}}}

\newcommand{\bR}{{\mathbb{R}}}

\newcommand{\bE}{{\mathbb{E}}}

\newcommand{\bC}{{\mathbb{C}}}

\renewcommand{\Re}{{\mathrm{Re}}}
\renewcommand{\Im}{{\mathrm{Im}}}
\newcommand{\prob}{{\mathrm{Pr}}}

\renewcommand{\epsilon}{\varepsilon}
\newcommand{\appropto}{\mathrel{\vcenter{
  \offinterlineskip\halign{\hfil$##$\cr
    \propto\cr\noalign{\kern2pt}\sim\cr\noalign{\kern-2pt}}}}}
\newcommand{\dket}[1]{\vert {#1} \rangle \! \rangle} 
\newcommand{\dbra}[1]{\langle \! \langle {#1} \vert} 
\definecolor{fluorescentpink}{rgb}{1.0, 0.08, 0.58}

\let\baraccent=\= % rename builtin command \= to \baraccent
\renewcommand{\=}[1]{\stackrel{#1}{=}} % for putting numbers above =
\newcommand{\thmref}[1]{\hyperref[#1]{Theorem~\ref{#1}}}
\newcommand{\lemmaref}[1]{\hyperref[#1]{Lemma~\ref{#1}}}
\newcommand{\propref}[1]{\hyperref[#1]{Proposition~\ref{#1}}}
\newcommand{\corollaryref}[1]{\hyperref[#1]{Corollary~\ref{#1}}}
\newcommand{\defref}[1]{\hyperref[#1]{Definition~\ref{#1}}}
\newcommand{\figref}[1]{\hyperref[#1]{Fig.~\ref{#1}}}
\newcommand{\tabref}[1]{\hyperref[#1]{Table~\ref{#1}}}
\newcommand{\figaref}[1]{\hyperref[#1]{Fig.~\ref{#1}(a)}}
\newcommand{\figbref}[1]{\hyperref[#1]{Fig.~\ref{#1}(b)}}
\newcommand{\figcref}[1]{\hyperref[#1]{Fig.~\ref{#1}(c)}}
\newcommand{\figdref}[1]{\hyperref[#1]{Fig.~\ref{#1}(d)}}
\newcommand{\figeref}[1]{\hyperref[#1]{Fig.~\ref{#1}(e)}}
\newcommand{\figfref}[1]{\hyperref[#1]{Fig.~\ref{#1}(f)}}
\renewcommand{\eqref}[1]{\hyperref[#1]{Eq.~(\ref{#1})}}
\newcommand{\secref}[1]{\hyperref[#1]{Sec.~\ref{#1}}}
\newcommand{\eqsref}[2]{\hyperref[#1]{Eqs.~(\ref{#1})-(\ref{#2})}}
\newcommand{\appref}[1]{\hyperref[#1]{Appx.~\ref{#1}}}

\begin{document}

\title{Quantum metrology of mixed states via purification}

\author{Sisi Zhou}\email{sisi.zhou26@gmail.com}
\affiliation{Perimeter Institute for Theoretical Physics, Waterloo, Ontario N2L 2Y5, Canada}
\affiliation{Department of Physics and Astronomy, Department of Applied Mathematics, and Institute for Quantum Computing, University of Waterloo, Ontario N2L 2Y5, Canada}

\date{\today}

\begin{abstract}
Multi-parameter quantum metrology is ubiquitous in physics, yet it remains fundamentally challenging. As reflected by the uncertainty principle, distinct unknown parameters generally cannot be estimated simultaneously at their individual ultimate precision limits. The quantum Cram\'{e}r--Rao bound (QCRB) and the Holevo Cram\'{e}r--Rao bound (HCRB) represent the ultimate precision bounds in multi-parameter quantum metrology. Here we introduce new formulations of these bounds based on quantum state purification. We show the bounds of a mixed state (under certain regularity conditions) are directly related to those of its purification, provided that nuisance parameters are introduced on the environmental system. Leveraging this connection, we develop a remarkably simple method for asymptotically attaining either the HCRB or twice the QCRB for mixed states under regularity conditions, using random purification channels that act collectively across state copies and subsequent measurements acting on individual copies.
\end{abstract}

\maketitle

\emph{Introduction.---}
Quantum metrolgy~\cite{degen2017quantum} aims to extract physical parameters from quantum systems to the highest precision. In many applications, the parameters of interest are not a single number, as in sensing of multiple components of magnetic fields, imaging, Hamiltonian estimation and noise characterization. Multi-parameter estimation problems are fundamentally different from the single-parameter counterparts. The measurement that is most informative about one parameter can be incompatible with the measurement that is most informative about another. As a result, the individual quantum limits for different parameters generally cannot all be reached simultaneously.

The Cram\'{e}r--Rao bound gives the ultimate precision bound for metrology, attainable through repeated sampling from parameter-dependent probability distributions. Single-parameter quantum metrology is analogous: a quantum state is measured many times and the quantum Cram\'{e}r--Rao bound (QCRB) is attainable by performing an optimal measurement~\cite{helstrom1967minimum,helstrom1968minimum,helstrom1976quantum,holevo2011probabilistic,braunstein1994statistical,barndorff2000fisher,paris2009quantum}. Multi-parameter metrology is fundamentally different due to measurement incompatibility~\cite{gill2000state,demkowicz2020multi,liu2019quantum,sidhu2020geometric,conlon2021efficient,ragy2016compatibility}. The ultimate precision limit is captured by the Holevo Cram\'{e}r--Rao bound (HCRB)~\cite{ragy2016compatibility,kahn2009local,yamagata2013quantum,yang2019attaining,albarelli2019evaluating}, which minimizes the weighted mean squared error over compatible measurements. Despite its theoretical power, the HCRB is difficult to attain~\cite{kahn2009local,yamagata2013quantum,yang2019attaining,tsang2026approaching}, where optimal measurements typically act collectively on multiple copies of the state with high implementation cost. This contrasts with the single-parameter case, where the QCRB and HCRB coincide and are efficiently attainable with measurements on individual copies~\cite{braunstein1994statistical}.

Fortunately, multi-parameter quantum metrology is greatly simplified for pure states. It was shown in~\cite{matsumoto2002new,gorecki2020optimal} that the HCRB of pure states can be attained using efficiently computable, individual measurements. Moreover, although the QCRB is generally not tight, it can also be attained up to a factor of two using individual measurements~\cite{hayashi1998asymptotic,li2016fisher,zhu2018universally,lu2025quantum}.  

Realistic metrological probes, however, are rarely pure. Decoherence entangles the system with unobserved environments, producing mixed states. 
A natural way to relate mixed and pure states is purification, which embeds a mixed state into a larger pure state by introducing an artificial environment. Our work bridges pure-state and mixed-state metrology by introducing purification-based formulations of the QCRB and HCRB. 
In particular, by introducing nuisance parameters that characterize the environmental degrees of freedom, we relate the precision bounds of a state purification to those of the original mixed state. 
This formulation has profound implications: it yields new measurement strategies that efficiently attain the HCRB, as well as twice the QCRB, using random purification channels that act collectively across state copies and subsequent individual measurements. Random purification channels~\cite{tang2025conjugate,girardi2025random,walter2025random,pelecanos2025mixed} transform multiple copies of a mixed state into multiple copies of its purification, with the environmental degrees of freedom randomized. They are known to have efficient circuit complexity that scales polynomially with the number of copies of the state~\cite{pelecanos2025mixed,bacon2006efficient,burchardt2025high}.  
After applying the random purification channel to all copies of the state, the rest of the measurement strategy follows a two-stage procedure (see \figref{fig:meas}): the first stage performs a coarse estimation of the parameters, including the nuisance parameters, and the second stage carries out optimal local estimation. 
By treating the environmental degrees of freedom introduced by purification as nuisance parameters, our formulation provides a systematic way to connect mixed-state precision bounds with their pure-state counterparts. This approach may also be useful for other noisy metrological tasks, including quantum channel estimation~\cite{fujiwara2008fibre,demkowicz2014using,gorecki2020optimal,hayashi2024finding,zhou2021asymptotic,mele2025optimal,chen2025quantum,yoshida2025random,girardi2025random2,chen2026quantum}.

\emph{Preliminaries.---}Consider a $d$-dimensional parameterized quantum state $\rho_{\theta}$ in Hilbert space $\mH_S$, where parameters are denoted by $\theta = (\theta_1,\theta_2,\ldots,\theta_m)$. 
$m$ is the number of parameters and $\Theta \subseteq \bR^m$ is the domain of $\theta$.
An estimator $\htheta(\ell)$ and the corresponding positive operator-valued measurement (POVM) $M = \{M_\ell\}_{\ell}$ is a function that maps the measurement outcomes $\ell$ to $\Theta$. Here $M_\ell$ are positive semidefinite measurement operators satisfying $\sum_\ell M_\ell = \id$.  The performance of an estimator at $\theta$ is characterized by the \emph{mean squared error matrix} (MSEM)
\begin{equation}
    V_{ij} = \sum_\ell (\htheta(\ell)_i - \theta_i)(\htheta(\ell)_j - \theta_j) \trace(\rho_\theta M_\ell). 
\end{equation}

\begin{figure}[tb]
    \centering
    \includegraphics[width=0.75\linewidth]{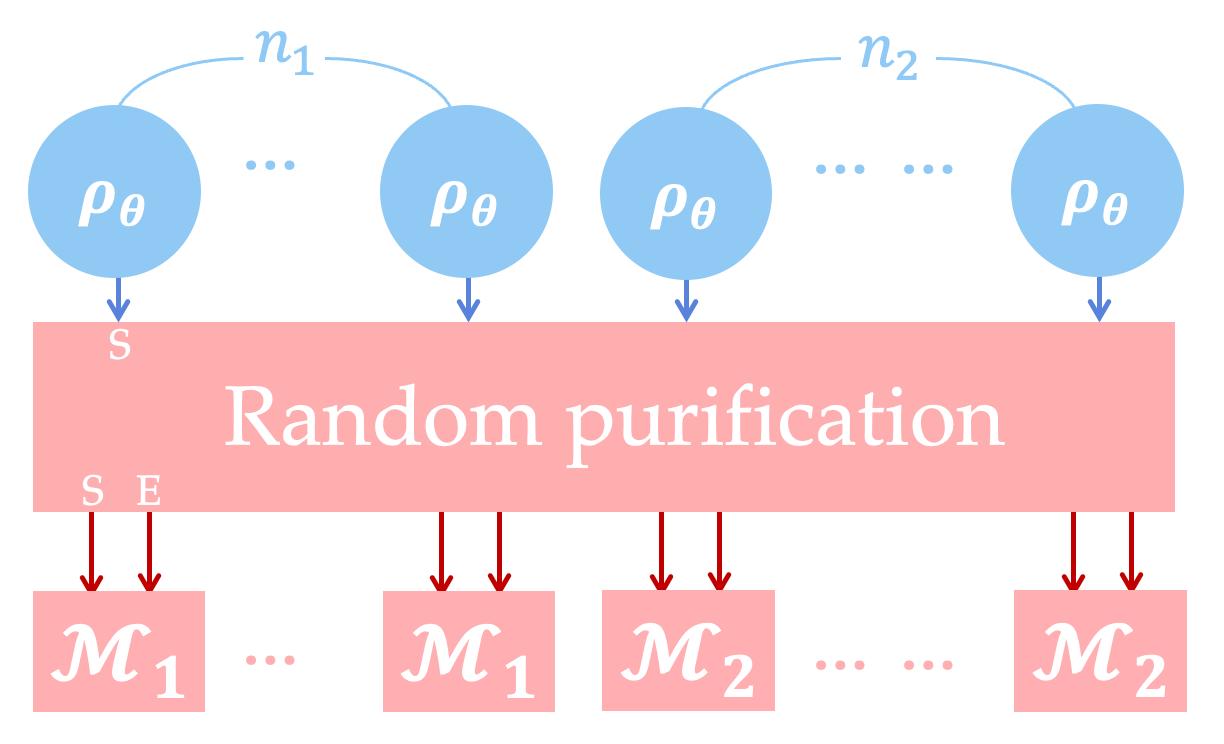}
    \caption{Measurement protocol. $S$ represents the probe system and $E$ represents the environmental system. $\dim(S)=d$ and $\dim(E) = r$, the rank of $\rho_\theta$. We first apply a random purification channel on $n=n_1+n_2$ copies of state $\rho_\theta$, which has circuit complexity that scales polynomially in $n$ and $\log(d)$. The outputs are their purifications with randomized environmental degrees of freedom. $\mM_1$ and $\mM_2$ correspond to a tomographic measurement, and a locally optimal measurement (for pure state estimation), respectively. }
    \label{fig:meas}
\end{figure}

The Cram\'{e}r--Rao bound (CRB)~\cite{rao1973linear,kay1993fundamentals,lehmann2006theory,cox2017inference,casella2002statistical} states that for any locally unbiased estimator at $\theta$ (i.e., estimator that is unbiased at $\theta$ and up to the first order in its vicinity), $V \succeq I(\rho_\theta, M)^{-1}$, where $A \succeq B$ means $A - B$ is positive semidefinite and $I(\rho_\theta, M)$ is the classical Fisher information matrix (CFIM). Given $n$ copies of $\rho_\theta$, the CRB is asymptotically attainable using the maximum likelihood estimator $\htheta^{(n)}_{\textsc{mle}}$~\cite{rao1973linear,kay1993fundamentals,lehmann2006theory,cox2017inference,casella2002statistical}---$\sqrt{n}(\hat{\theta}^{(n)}_{\textsc{mle}} - \theta)$ converges in distribution to a normal distribution centered around $0$ with variance $I(M)^{-1}$. 
The QCRB~\cite{helstrom1967minimum,helstrom1968minimum,helstrom1976quantum,holevo2011probabilistic,braunstein1994statistical,barndorff2000fisher,paris2009quantum} further states that, for any locally unbiased estimator at $\theta$, 
\begin{equation}
    V \succeq J(\rho_\theta)^{-1}, 
\end{equation}
where $J(\rho_\theta)$ is the QFIM. We assume $\theta$ is identifiable in $\Theta$, i.e. $\rho_\theta$ is injective, and $J(\rho_\theta) \succ 0$ within $\Theta$ throughout this work. 

For single-parameter estimation, the QFIM is equal to the CFIM optimized over all POVMs~\cite{braunstein1994statistical}. 
For multi-parameter estimation, however, this does not always hold. It helps to consider the \emph{weighted mean squared error} (WMSE), $\trace(W V)$, where $W \succeq 0$ is a cost matrix. The QCRB immediately provides a lower bound on the WMSE: $\trace(W V) \geq \trace(W J(\rho_\theta)^{-1}) =: C_{\rm F}(\rho_\theta,W)$. It can be expressed using the following optimization problem~\cite{ragy2016compatibility}: 
\begin{equation}
\label{eq:F-formula}
       C_{\rm F}(\rho_\theta,W) = \min_{\vX = \{X_i\}_{i=1}^m} \trace(W \Re[Z(\vX)]) ,    
\end{equation}
where $X_i$ are Hermitian operators acting on $\mH_S$ satisfying $\forall i,j \in [m]$, 
\begin{equation}
\label{eq:cond}
    \trace\bigg(X_i \frac{\partial \rho_\theta}{\partial \theta_j}  \bigg) = \delta_{ij}, \quad Z(\vX)_{ij} = \trace(\rho_\theta X_i X_j). 
\end{equation} 
The HCRB~\cite{holevo2011probabilistic,kahn2009local,yamagata2013quantum,yang2019attaining} provides a tighter lower bound: 
\begin{multline}
     C_{\rm H}(\rho_\theta,W):=  \min_{\vX:=\{X_i\}_{i=1}^m} \!\trace(W \Re[Z(\vX)]) \\+  \big\|\sqrt{W}  \Im[Z(\vX)]  \sqrt{W}\big\|_1, 
\end{multline}
where $X_i$ are Hermitian operators satisfying \eqref{eq:cond}. 
It was known $C_{\rm F}(\rho_\theta,W) \leq C_{\rm H}(\rho_\theta,W) \leq 2 C_{\rm F}(\rho_\theta,W)$~\cite{carollo2019quantumness,tsang2020quantum}. The HCRB is asymptotically attainable through collective measurements acting on $\rho_\theta^{\otimes n}$ when the state and estimators satisfy certain regularity conditions~\cite{kahn2009local,yamagata2013quantum,yang2019attaining,tsang2026approaching}. However, reducing the implementation and computational complexity of the optimal measurement strategy remains significant challenges. 

For pure states, the optimal measurements act locally on individual copies, and the computational complexity to find them scales polynomially with $d$~\cite{matsumoto2002new}. Moreover, if our goal is to attain twice the QCRB---where the factor of 2 is the optimal constant factor in general---we can apply individual Haar random measurements on each copy~\cite{hayashi1998asymptotic,lu2025quantum} or two-design random measurements on every two copies, both choices independent of probe states~\cite{li2016fisher,zhu2018universally}. Finally, it is worth noting that using improved quantum state tomography~\cite{pelecanos2025debiased,pelecanos2025mixed}, it was recently shown there exists a locally unbiased estimator (as a function of $\theta$) that attains twice the QCRB for large $n$. However, the asymptotic convergence to twice the QCRB with no prior knowledge of $\theta$ was not directly implied.

\emph{Purification-based formulation of QCRB and HCRB.---}
Here we introduce our main theorem. We first introduce the purification of parameterized states we used in this work. Let $\rho_\theta = \sum_{i=1}^r \lambda_{i,\theta}\ket{e_{i,\theta}}\bra{e_{i,\theta}}$ be the diagonalization of $\rho_\theta$ where $r$ is the rank of the state, fixed and known in advance, and the eigenvalues $\lambda_{i,\theta} > 0$ are distinct for $i\leq r$.
We consider its parameterized purification as 
\begin{equation}
        \ket{\psi_{\theta,\varphi}} := \sum_{j=1}^r  \sqrt{\lambda_{j,\theta}}\ket{e_{j,\theta}}_S \ket{f_{j,\varphi}}_E. 
\end{equation}
We use the subscripts $_S$ and $_E$ to denote the probe system and the environmental system. $\mH_E$ is a $r$-dimensional space and $\{ \ket{f_{j,\varphi}}\}_{j=1}^r$ is an orthonormal basis. $\varphi$ are called \emph{nuisance parameters} because their values are not of interest.

For $\rho_\theta$ satisfying certain regularity conditions (as detailed in~\cite{SM}\nocite{LeeSmoothManifolds}), we have the following theorem.  
\begin{theorem}
\label{thm:purification}
Assume $\varphi \in \Omega$ where $\Omega$ is an open, bounded set and $\frac{\partial}{\partial \varphi_i} \ket{f_{j,\varphi}} = -i H_{i,\varphi}\ket{f_{j,\varphi}}$ where $\{H_{j,\varphi}\}_{j=1}^{r^2-1}$ form a basis of traceless Hermitian operators on $\mH_E$ for any $\varphi \in \Omega$. 
Given a cost matrix $W \in \bR^{m \times m}$ satisfying $W \succ 0$, for some fixed $(\theta,\varphi) \in \Theta \times \Omega$, 
\begin{gather}
\label{eq:qfi-inv}
    J(\rho_\theta)^{-1} = ( J(\psi_{\theta,\varphi})^{-1} )_{SS},\\
    C_{\rm H}(\rho_\theta,W) = C_{\rm H}(\psi_{\theta,\varphi},W^*).
\end{gather}
\end{theorem}
\noindent Here we abuse the notation a bit and use $_S$ to denote the matrix indices corresponding to $\theta$. $J(\psi_{\theta,\varphi})$ and $W^* \in \bR^{(m + r^2 -1) \times (m + r^2 -1)}$, where the first $m$ indices correspond to $\theta$ and the rest correspond to $\varphi$. $J(\psi_{\theta,\varphi})_{SS} \in \bR^{m  \times m}$ is equal to $J(\psi_{\theta,\varphi})$ when restricted to the first $m$ indices; and $W^* \in \bR^{m+r^2-1 \times m+r^2-1}$ is equal to $W$ in the upper-left block and zero in the rest.

To simplify the notation, let $m^* = m + r^2 -1$ and 
\begin{equation}
\phi := (\theta,\varphi) \in \bR^{m^*}
\end{equation}
represent all unknown parameters in the parameterized purification $\psi_{\phi} := \psi_{\theta,\varphi}$. 

\begin{proof}
We first prove $J(\rho_\theta)^{-1} = ( J(\psi_{\theta,\varphi})^{-1} )_{SS}$. We only need to show $C_{\rm F}(\rho_\theta,W) = C_{\rm F}(\psi_{\phi},W^*)$ for any $W \in \bR^{m\times m}$. 
Note that 
    \begin{equation}
    \label{eq:F-formula-1}
        C_{\rm F}(\psi_{\phi},W^*) = \min_{\vX = \{X_i\}_{i =1}^{m^*}} \trace(W^* \Re[Z(\vX)]) .
    \end{equation}
    where $Z(\vX)_{ij} = \trace(\psi_\phi X_i X_j)$ and $X_i \in \bC^{dr \times dr}$ are Hermitian operators satisfying $\forall i,j \in [m^*]$, 
    \begin{equation}
\label{eq:cond-2}
    \trace\bigg(X_i \frac{\partial \psi_\phi}{\partial \phi_j}  \bigg) = \delta_{ij}.
\end{equation} 
Note that the QCRB can be formulated using \eqref{eq:F-formula-1} and \eqref{eq:cond-2} only when $J(\psi_\phi) \succ 0$. We prove $J(\psi_\phi) \succ 0$ in \cite{SM}.

First, we show $\trace(W^*  J(\psi_{\phi})^{-1}) \leq \trace(W J(\rho_\theta)^{-1})$. We need to show first any $\{X_i\}_{i\in[m]}$ satisfying \eqref{eq:cond} can be turned into $\{X^\diamond_i\}_{i \in [m^*]}$ satisfying \eqref{eq:cond-2}. For $i \in [m]$, $(X_i)_S \otimes \id_E = (X^\diamond_i)_{SE}$. Note that $\trace((X_i \otimes \id_E)\frac{\partial \psi_\phi}{\partial \theta_j}) = \trace(X_i \frac{\partial \rho_\theta}{\partial \theta_j}) = \delta_{ij}$ and $\trace((X_i \otimes \id_E)\frac{\partial \psi_\phi}{\partial \varphi_j}) = \trace(X_i \frac{\partial \rho_\theta}{\partial \varphi_i}) = 0$.  The rest of the Hermitian operators $\{(X^\diamond_i)_{SE}\}_{i \in [m^*]\backslash[m]}$ can be chosen as any solution satisfying \eqref{eq:cond-2}, which always exists because $\{\frac{\partial \psi_\phi}{\partial \phi_j}\}_{j\in[m^*]}$ are linearly independent~\cite{SM}. Finally, we note that $Z(\vX^\diamond)_{SS} = Z(\vX)$, which proves $\trace(W^*  J(\psi_{\phi})^{-1}) \leq \trace(W J(\rho_\theta)^{-1})$. 

Next, we show $\trace(W^*  J(\psi_{\phi})^{-1}) \geq \trace(W J(\rho_\theta)^{-1})$. We need to show first any $\{X_i\}_{i\in[m^*]}$ satisfying \eqref{eq:cond-2} can be turned into $\{\tilde{X}_i\}_{i \in [m]}$ satisfying \eqref{eq:cond}. 
Assume Hermitian $X_i \in \bC^{dr \times dr}$ satisfies \eqref{eq:cond-2}. Let $\tilde{X}_i \in \bC^{d \times d}$ be operators satisfying  
\begin{equation}
    (\tilde{X}_i \otimes \id_E) \ket{\psi_{\phi}} = X_i \ket{\psi_{\phi}}, \label{eq:tilde-x}
\end{equation}
for all $i \in [m]$, which are not necessarily Hermitian at this moment. $\tilde{X}_i$ always exists because the Schmidt rank of $\psi_{\phi}$ is $r$. Let $\Pi$ be the projection onto the support of $\rho_\theta$ and $\Pi_\perp = \id - \Pi$. 
We first notice we can always assume $\tilde{X_i}$ satisfies $\Pi_\perp  \tilde{X_i} \Pi_\perp = 0$ and $\Pi \tilde{X_i} \Pi_\perp = ( \Pi_\perp \tilde{X_i} \Pi )^\dagger$ because \eqref{eq:tilde-x} is independent of any change in $\tilde{X_i}\Pi_\perp$ and we can choose this part freely. 

Below we use the vectorization formalism $\dket{\star} = \sum_{j\in[d],k\in[r]} (\star) \ket{e_{j,\theta}}\ket{f_{k,\varphi}}$ (although only $\{\ket{e_{j,\theta}}\}_{i\in[r]}$ was defined but it can be easily enlarged to be an entire orthonormal basis). Then $\ket{\psi_{\phi}} = \dket{R}$ for $R \in \bR^{d\times r}$ where $R_{ij} = \sqrt{\lambda_{i,\theta}}\delta_{ij}$ when $i \in [r]$ and $R_{ij} = 0$ when $i \in [d]\backslash[r]$. Note that 
    \begin{align}
    \frac{\partial \psi_{\phi}}{\partial \varphi_j} 
    &= (-i \id_S \otimes H_{j})\ket{\psi_{\phi}}\!\bra{\psi_{\phi}} + \ket{\psi_{\phi}}\!\bra{\psi_{\phi}}(i \id_S \otimes H_{j})\nonumber \\
    &= -i \dket{RH_{j}^\top }\dbra{R} + i \dket{R}\dbra{RH_{j}^\top } . 
    \end{align}
    From $\trace(X_i \frac{\partial \psi_{\phi}}{\partial \varphi_j}) = 0$, 
    \begin{gather}
    -i \trace(  R^\top \tilde{X}_i^\dagger R H_{j}^\top  ) + i \trace( R^\top \tilde{X_i} R H_{j}^\top  ) = 0,\\
    \trace( (-i \tilde{X}_i^\dagger + i \tilde{X}_i) R H_{j}^\top R^\top ) = 0. 
    \end{gather}
    Here $\big\{R H_{j}^\top  R^\top\big\}_{j=1}^{r^2-1}$ span the entire space of Hermitian operators in subspace $\Pi$ that is perpendicular to $\rho_\theta^{-1}$ (here $^{-1}$ means pseudo-inverse). We must have 
    \begin{equation}
    \Pi (-i \tilde{X}_i^\dagger + i \tilde{X}_i) \Pi = a_i \rho_\theta^{-1},\quad \text{ for some } a_i \in \bR.  \label{eq:herm}        
    \end{equation}
    From \eqref{eq:tilde-x}, $
        \bra{\psi_\phi}(\tilde{X}_i \otimes \id_E) \ket{\psi_{\phi}} = \bra{\psi_{\phi}} X_i \ket{\psi_{\phi}} \in \bR$. 
    It implies we must have $a_i = 0$ and then $\tilde{X}_i$ is Hermitian. Note that $\trace\big(X_i \frac{\partial \psi_{\phi}}{\partial \theta_j}\big) = \trace\big(\tilde{X}_i \frac{\partial \rho_{\theta}}{\partial \theta_j}\big) = \delta_{ij}$. Therefore, we have constructed a set of Hermitian operators $\tilde{X_i}$ that satisfies \eqref{eq:cond} from $X_i$ satisfying \eqref{eq:cond-2}. Finally, we note that $Z(\tilde\vX) = Z(\vX)_{SS}$, which proves $\trace(W^*  J(\psi_{\phi})^{-1}) \geq \trace(W J(\rho_\theta)^{-1})$.  

    We have shown above $C_{\rm F}(\rho_\theta,W) = C_{\rm F}(\psi_{\phi},W^*)$ for any $W \in \bR^{m\times m}$ which implies $J(\rho_\theta)^{-1} = ( J(\psi_{\theta,\varphi})^{-1} )_{SS}$. $C_{\rm H}(\rho_\theta,W) = C_{\rm H}(\psi_{\phi},W^*)$ also holds true straightforwardly because the HCRB can also be written as an optimization over Hermitian operators satisfying \eqref{eq:cond}, $Z(\vX^\diamond)_{SS} = Z(\vX)$ when $\vX$ satisfies \eqref{eq:cond} and $Z(\tilde{\vX}) = Z(\vX)_{SS}$ when $\vX$ satisfies \eqref{eq:cond-2}. 
    \end{proof}

Note that purification-based techniques were widely applied before, such as in state tomography~\cite{pelecanos2025mixed}, testing~\cite{chen2024local}, single-parameter metrology~\cite{fujiwara2008fibre,escher2011general,demkowicz2012elusive}, and also in multi-parameter quantum metrology~\cite{albarelli2022probe,pelecanos2025mixed}. 
However, our nuisance-parameter treatment is distinct from previous applications.

Traditionally, another purification-based method was widely applied to calculate the QFIM, where the QFIM of a mixed state is obtained by minimizing the QFIM over all its $\theta$-dependent purifications. It was widely studied in single-parameter estimation~\cite{fujiwara2008fibre,escher2011general,demkowicz2012elusive} and generalized to multi-parameter estimation recently~\cite{albarelli2022probe}. It states
\begin{equation}
    J(\rho_{\theta}) = \min_{\varphi(\theta)} J(\psi_{\theta,\varphi(\theta)}),
\end{equation}
where $\varphi(\theta)$ is any smooth function of $\theta$ and the matrix minimization is understood with respect to the partial order $\preceq$. 
To see its connection to our formulation, writing the QFIM of $\psi_{\theta,\varphi}$ in block form as $$J(\psi_{\theta,\varphi}) = \begin{pmatrix}
    J(\psi_{\theta,\varphi})_{SS} & J(\psi_{\theta,\varphi})_{SE}\\
    J(\psi_{\theta,\varphi})_{ES} & J(\psi_{\theta,\varphi})_{EE}
\end{pmatrix},$$ $J(\rho_\theta)$ from \eqref{eq:qfi-inv} is equal to the Schur complement $$J(\psi_{\theta,\varphi})_{SS} - J(\psi_{\theta,\varphi})_{SE} J(\psi_{\theta,\varphi})_{EE}^{-1} J(\psi_{\theta,\varphi})_{ES}.$$  $J(\psi_{\theta,\varphi})_{SS}$ is the QFIM of purification $\psi_{\theta,\varphi}$ when $\varphi$ is independent of $\theta$, which is an upper bound on $J(\rho_\theta)$ because observing $SE$ can in general provide more information about $\theta$. The additional term $J(\psi_{\theta,\varphi})_{SE} J(\psi_{\theta,\varphi})_{EE}^{-1} J(\psi_{\theta,\varphi})_{ES}$ represents exactly the amount of overestimate of the QFIM. In fact, the Schur complement is precisely the minimum of the pure-state QFIM over all choices of the environmental gauge $\varphi=\varphi(\theta)$. In our formulation, the minimization problem is recast as the QFIM inversion problem by converting the environmental freedom to nuisance parameters independent of $\theta$. A detailed comparison between our method and the traditional purification method can be found in~\cite{SM}.

\emph{Attaining the HCRB and twice the QCRB.---}
Here we present an efficient measurement strategy that attains the HCRB or twice the QCRB. For $\rho_\theta$ satisfying certain regularity conditions (as detailed in~\cite{SM}), we have the following theorem. 
\begin{theorem}
\label{thm:hcrb}
Given a cost matrix $W \in \bR^{m \times m}$ satisfying $W \succ 0$, a measurement protocol using random purification channels acting collectively across $n$ copies of $\rho_\theta$ and subsequent individual measurements yields an asymptotically unbiased estimator $\htheta^{(n)}$\footnote{An asymptotically unbiased estimator is an estimator satisfying $\forall \theta \in \Theta,\lim_{n\rightarrow \infty} \bE[\htheta^{(n)}] = \theta$.}, such that the WMSE satisfies $\forall \theta \in \Theta$, 
\begin{gather}
\label{eq:hcrb-cov}
    \lim_{n\rightarrow \infty}n \trace(WV(\htheta^{(n)})) = C_{\rm H}(\rho_\theta,W). 
\end{gather}
Switching to different individual measurements yields another asymptotically unbiased estimator such that the MSEM satisfies $\forall \theta \in \Theta$,
\begin{equation}
\label{eq:qcrb-cov}
       \lim_{n\rightarrow \infty}n V(\htheta^{(n)}) = 2 J(\rho_\theta)^{-1}.  
\end{equation}
\end{theorem}

The rigorous proof of \thmref{thm:hcrb} is provided in End Matter, with additional proof details included in~\cite{SM}. We give an informal description of our measurement protocol here. Random purification channels are channels that purify multiple copies of mixed states with the environmental degrees of freedom randomized, i.e.~\cite{tang2025conjugate,girardi2025random,walter2025random,pelecanos2025mixed},  
\begin{equation}
    \Phi_n^{d,r}(\rho_S^{\otimes n}) = \bE_{U_E} [\big( (\id_S \otimes U_E) \psi_{SE} (\id_S \otimes U_E^\dagger) \big)^{\otimes n}],
\end{equation}
where $\bE_{U_E}$ denotes expectation with respect to the Haar measure on the group of $r \times r$ unitary matrices and $\psi_{SE}$ is an arbitrary purification of $\rho_S$. Our measurement protocol requires first applying a global random purification channel on all $n$ copies of state $\rho_{\theta}$. 

Afterwards, we use a two-stage method. In Stage~1, we take $n_1 \ll n$ copies of the output pure state and perform state tomography. We perform measurements on individual copies of states. For concreteness, we take a 3-design measurement $\{M_{\bar{s}} = q_{\bar{s}} \ket{\bar{s}}\!\bra{\bar{s}}\}$ where $q_{\bar{s}}$ is a normalization constant, following from the classical shadow learning strategy. 
In principle, any individual measurement that provides an \emph{unbiased, bounded state estimator} is sufficient for our purpose, as can be seen from our proof below. It is because in the asymptotic limit $n_1/n \rightarrow 0$, the contribution from Stage~1 will be negligible and the measurement choice is not important. Stage~1 provides a rough estimate of $\theta$ and $U_E$. 

In Stage~2, we consider the purification $\psi_{\theta,U}$ as $\psi_{\theta,\varphi}$, a function of $\theta$ and nuisance parameters $\varphi$ (\figref{fig:mesh}). Starting from an rough estimate $\psi_{\ctheta,\cvarphi} = \psi_{\cphi}$, to attain the HCRB, we take the other $n_2 = n-n_1$ copies of the output pure state and perform the optimal locally unbiased estimation at $\cphi$~\cite{matsumoto2002new}. It provides an optimal rank-one individual measurement (depending on $\cphi$) attaining the HCRB for pure states, denoted by $\{M_{\bar{\ell}} = \ket{b_{\bar{\ell}}}\!\bra{b_{\bar{\ell}}}\}$ (for notation simplicity, $\ket{b_{\bar{\ell}}}$ is unnormalized). Stage~1 guarantees the closeness between $\phi$ and $\cphi$, which in turn ensures a bounded estimation bias and WMSE in Stage~2 that converge to $0$ and $1/n\cdot C_{\rm H}(\psi_{\theta,\varphi},W^*)$ as $n \rightarrow \infty$, the latter equal to $1/n\cdot C_{\rm H}(\psi_{\theta},W)$ from \thmref{thm:purification}. To attain twice the QCRB, we instead perform a different locally optimal measurement, called Fisher-symmetric measurement, that obtains a CFIM equal to half the QFIM (at $\phi = \cphi$)~\cite{li2016fisher}. Other simpler measurements might also work~\cite{zhou2026randomized}, if a larger constant gap to the QCRB is allowed.
The MSEM for $\theta$-estimation converges to $2/n \cdot (J(\psi_{\theta,\varphi})^{-1})_{SS}$, which is equal to $2/n \cdot J(\rho_{\theta})^{-1}$ from \thmref{thm:purification}.

\emph{Discussion.---}
A mixed state can be viewed as part of a larger pure state in a composite system, called a purification, whose environmental degrees of freedom are unobserved. In parameter estimation, if the entire purification is known and accessible, it can contain more information about the parameters than the mixed state. \thmref{thm:purification} shows the loss of information associated with tracing out the environment can be evaluated by preserving the purification structure while treating its environmental degrees of freedom as unknown nuisance parameters. In other words, estimating parameters from a mixed state (under certain regularity conditions) is metrologically equivalent, at the level of both the QCRB and HCRB, to estimating them from a purification while requiring insensitivity to environmental variations. 

\thmref{thm:purification} can also be useful computationally. When a purification of a mixed state is directly available, the metrological bounds on the mixed state can be calculated using \thmref{thm:purification} without diagonalization of the mixed state. \thmref{thm:purification} also helps convert results from unitary (noiseless) channel estimation to general channel estimation. We provide one example in \cite{SM}, where we show how an efficient semidefinite program for estimating parameters in unitary channels can be converted into another one for general channel estimation.

\thmref{thm:hcrb} is another application of \thmref{thm:purification} where mixed state estimation is converted to simpler pure state estimation via random purification channels. 
It provides a measurement protocol with circuit complexity that scales polynomially with $n$ to asymptotically attain the HCRB and twice the QCRB. It is worth investigating how to drop some of the regularity requirements in future, including the fixed-rank assumption, and the nondegenerate positive spectrum assumption, via pre-estimation steps.

\emph{Note.---}During the preparation of this manuscript, we became aware of Ref.~\cite{tsang2026approaching}, which proposes a different measurement strategy for attaining the HCRB. 

\emph{Acknowledgments.---}
We would like to thank Senrui Chen, Qiushi Liu and Beni Yoshida for helpful discussions. S.Z. acknowledges funding provided by Perimeter Institute for Theoretical Physics, a research institute supported in part by the Government of Canada through the Department of Innovation, Science and Economic Development Canada and by the Province of Ontario through the Ministry of Colleges and Universities. We acknowledge the use of ChatGPT in checking the proofs and addressing minor gaps. The author bears sole responsibility for the correctness of the proofs and results.

\bibliography{refs}

\onecolumngrid
\newpage
\onecolumngrid
\begin{center}
\large\textbf{End Matter}
\end{center}
\vspace{0.1in}
\twocolumngrid

\emph{End Matter.---} 
Below we prove \thmref{thm:hcrb}. $\norm{\cdot}_{p}$ represents the Schatten-$p$ norm of matrices and $\norm{\cdot} = \norm{\cdot}_\infty$. 
\begin{proof}
We first prove \eqref{eq:hcrb-cov}. Let $\ket{\psi_{\theta,U}} := \sum_{j=1}^r  \sqrt{\lambda_{j,\theta}}\ket{e_{j,\theta}}_S U\ket{j}_E$, where $\ket{j}$ is any fixed orthonormal basis in $E$. The output state after the random purification channel is $\bE_U [\psi_{\theta,U}^{\otimes n}]$. 

In Stage~1, we perform 3-design measurements $\{q_{s_k}\ket{s_k}\!\bra{s_k}\}$ on the $k$-th copy for $k \in [n_1]$, with measurement outcome $s = (s_1,\ldots,s_{n_1})$. We also let $q_s = \prod_{k=1}^{n_1} q_{s_k}$. The state on the remaining $n_2$ registers is $\rho_{s|\theta}$ and 
\begin{equation}
    \prob(s|\theta) \rho_{s|\theta} = \bE_{U} [ \braket{s|\psi_{\theta,U}^{\otimes n_1}|s} \psi_{\theta,U}^{\otimes n_2} ]. 
\end{equation}
We then parameterize the environmental system and design Stage~2, according to the outcome $s$. Let $\cpsi(s) = \frac{1}{n_1}\sum_{k=1}^{n_1} \cpsi^{(k)}(s_k) = \frac{1}{n_1}\sum_{k=1}^{n_1}((dr+1)\ket{s_k}\!\bra{s_k} - \id)$ be the classical shadow state estimator---which is not necessarily pure or positive. 
Then we take $(\ctheta,\cU)$ to be one satisfying 
\begin{equation}
\textstyle
\label{eq:dist-1}
    \big\| \psi_{\ctheta,\cU} - \cpsi(s) \big\|_1 \leq 2\inf_{\theta,U}  \big\|  \psi_{\theta,U} - \cpsi(s) \big\|_1. 
\end{equation}
Note that $\ctheta,\cU$ depend on $s$ but we do not show it explicitly for notation simplicity.

We then introduce the following parameterization of the environmental system: 
\begin{equation*}\textstyle
    \ket{\psi_{\theta,\varphi}} := \sum_{j=1}^r  \sqrt{\lambda_{j,\theta}}\ket{e_{j,\theta}}_S \exp(-i\sum_{i=1}^{r^2-1} \varphi_i H_i)\cU\ket{j}_E,
\end{equation*}
where $H_i$ is some fixed basis of traceless Hermitian operators. Let $\cvarphi = 0$. Then we choose the optimal locally unbiased estimator $\htheta^{(k)}(s,\ell_k)$ at the local point $(\theta,\varphi) = (\ctheta,\cvarphi) = (\ctheta,0)$ from Matsumoto's HCRB~\cite{matsumoto2002new}. $\ket{b_{\ell_k}}\bra{b_{\ell_k}}$ is the rank-one measurement operator and $\ell_k$ is the measurement outcome on the $k$-th copy (among the $n_2$ copies). $\ell = (\ell_{1},\ldots,\ell_{n_2})$. We take the final estimator of $\theta$ to be $\htheta(s,\ell) = \frac{1}{n_2}\sum_{k=1}^{n_2}\htheta^{(k)}(s,\ell_k)$. 
Then the MSEM is
\begin{align*}
\textstyle
    V_{ij} &= \!\sum_{s,\ell} q_s (\htheta(s,\ell)_i-\theta_i) (\htheta(s,\ell)_j-\theta_j) \prob(s|\theta)\! \bra{b_\ell}\!\rho_{s|\theta}\!\ket{b_\ell}\\
    &= \!\sum_{s,\ell} q_s (\htheta(s,\ell)_i-\theta_i) (\htheta(s,\ell)_j-\theta_j)  
    \\* &\qquad \qquad \qquad \qquad\quad \cdot 
    \bE_{U} [ \braket{s|\psi_{\theta,U}^{\otimes n_1}|s} \!\bra{b_\ell}\!\psi_{\theta,U}^{\otimes n_2}\!\ket{b_\ell} ],
\end{align*}
where $\ket{b_\ell} = \bigotimes_{k=1}^{n_2}\ket{b_{\ell_k}}$.

\begin{figure}[tb]
    \centering
    \includegraphics[width=0.6\linewidth]{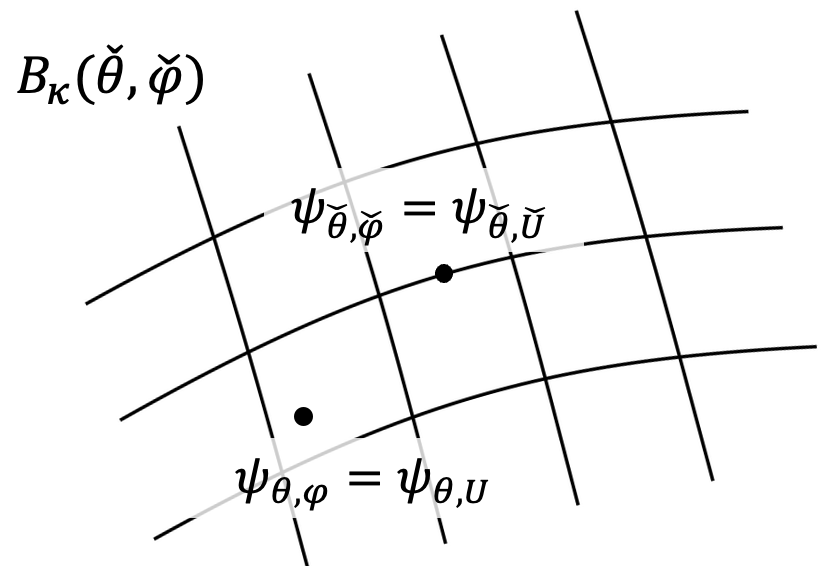}
    \caption{Parameterization. When the distance between the rough estimate ${\psi_{\ctheta,\cU}} =: {\psi_{\ctheta,\cvarphi}}$ and the true state ${\psi_{\theta,U}}$ is sufficiently small, there is a $\kappa$-ball $B_\kappa(\cphi)$ centered at $\cphi = (\ctheta,\cvarphi)$ such that inside $B_\kappa(\cphi)$, any $\psi_{\theta,U}$ can be represented by $\psi_{\theta,\varphi}$ for a unique choice of $\varphi$ and $\psi_{\theta,\varphi}$ is an injective function. }
    \label{fig:mesh}
\end{figure}

Next, we show $V_{ij}$ converges as expected. Let $T = T_{U,s,n_1}$ denote an event where 
\begin{equation}
\label{eq:dist-2}
\big\| \psi_{\theta,U} - \cpsi(s) \big\|_1 \leq 1/n_1^{1/2-\delta/2}
\end{equation}
for some small $\delta$, e.g. we can take $\delta = 0.1$. We analyze the MSEM according to whether this event occurs. First, we note that 
\begin{align*}
    1 - \prob(T_{U,s,n_1}) 
    &\leq \prob( dr \|\psi_{\theta,U} - \cpsi(s)\| \geq 1/n_1^{1/2-\delta/2})\\
    & \leq 2dr \exp\big( - {n_1^{\delta}}/(8(dr)^2(dr+1)^2)\big),
\end{align*}
where we use the Matrix Hoeffding's inequality~\cite{tropp2012user} for any random $dr \times dr$ Hermitian matrices $Q_k$ with $\bE[Q_k] = 0$: 
\begin{equation}
\textstyle
    \prob(\|\sum_k Q_k\| \geq t) \leq 2dr e^{-t^2/8q^2},\;\norm{\sum_k Q_k^2}\leq q^2, 
\end{equation}
$\bE_{s_k}[\psi_{\theta,U} - \cpsi^{(k)}(s_k)] = 0$, and $\big\| \psi_{\theta,U} - \cpsi^{(k)}(s_k)) \big\| \leq dr + 1$. 
Thus, $T_{U,s,n_1}$ almost always occurs as $n_1 \rightarrow \infty$.

When $T$ occurs, we first justify the validity of the parameterization step. We show~\cite{SM}, there exist universal constants $\kappa, \eta > 0$, such that $1 - |\!\braket{\psi_{\theta,U}|\psi_{\ctheta,\cU}}\!|^2 \leq \eta$ implies $\psi_{\theta,U}$ can be represented as $\psi_\phi$ by some $\phi \in B_{\kappa}(\cphi):=\{\phi|\|\phi-\cphi\| < \kappa\}$ and $\psi_{\phi}$ on $B_{\kappa}(\cphi)$ is an injective map. 
Conditioned on $T$ occurring, from \eqref{eq:dist-1}, \eqref{eq:dist-2} and the triangle inequality, we have $\|\psi_{\theta,U} - \psi_{\ctheta,\cU}\|_1 \leq  3 n_1^{-1/2+\delta/2}$ which implies $1 - |\!\braket{\psi_{\theta,U}|\psi_{\ctheta,\cU}}\!|^2  \leq \frac{9}{4} n_1^{-1+\delta}$. 
For sufficiently large $n_1$, $\frac{9}{4} n_1^{-1+\delta} < \eta$, and $(\theta,\varphi)$ will then be identifiable within $B_\kappa(\cphi)$. 
In addition, for large enough $n_1$, $\norm{\Delta\phi} = O(n_1^{-1/2+\delta/2})$ that converges uniformly independent of $\phi$ and $\cU$~\cite{SM}. Here $\Delta \phi := \phi - \cphi$.  
Next, we calculate $( V_{U,s})_{ij} := 
\sum_{\ell} (\htheta(s,\ell)_i-\theta_i) (\htheta(s,\ell)_j-\theta_j)  \!\bra{b_\ell}\!\psi_{\theta,U}^{\otimes n_2}\!\ket{b_\ell}$. We use the notation $V_{U,s}$ to denote the MSEM averaged over $\ell$ with $(U,s)$ fixed. When $T$ occurs, $\htheta^{(k)}(s,\ell_k)_i$ is bounded, and satisfies~\cite{SM} 
\begin{equation}
\label{eq:sm-1}
    \bE_{\ell_k} [\htheta^{(k)}(s,\ell_k)_i] = \theta_i + O(\|\Delta\phi\|^2),
\end{equation}
where the bias is due to the deviation of the rough estimate $\cphi$ from the true value of $\phi$. We also have~\cite{SM} 
\begin{multline}
\label{eq:sm-2}
    \trace( W  V_{U,s}) = \\ \frac{1}{n_2}\big(C_{\rm H}(\psi_{\ctheta,\cvarphi},W^*) + O(\|\Delta\phi\|)\big) + O(\|\Delta\phi\|^4). 
\end{multline}
In particular, $C_{\rm H}(\psi_{\ctheta,\cvarphi},W^*)$ can be replaced by $C_{\rm H}(\psi_{\ctheta},W)$ using \thmref{thm:purification}. 

Finally, 
\begin{equation}
    \trace(W V) = \bE_{U,s}[\trace( W  V_{U,s} )|T] \,\prob(T)  +O(\prob(T^c)), 
\end{equation}
where the first and second terms correspond, respectively, to the cases where $T$ occurs and where it does not (represented by $T^c$). The second term is bounded by $O(\prob(T^c))$ because our estimator is uniformly bounded~\cite{SM}. To take the limit $n \rightarrow \infty$ of $n\trace(W V)$, we want the second term $n O(\prob(T^c))$ to vanish and the first term $n \bE_{U,s}[\trace( W  V_{U,s} )|T]$ to converge to the HCRB. The second term is easy to control because $\prob(T^c)$ decays exponentially. To ensure that the first term converges, we first note that we can interchange $\bE_{U,s}$ and the limit $n \rightarrow \infty$ because $\Delta\phi \rightarrow 0$ uniformly. Next, we need to pick a sufficiently large $n_1$ such that both terms in \eqref{eq:sm-2} converge when multiplied by $n$. One valid choice is $n_1 = n^{{2}/{(3(1-\delta))}}$ and $n_2 = n - n_1$. Finally, using the continuity of the HCRB~\cite{SM}, we have $\lim_{n\rightarrow \infty} n \trace(W V) = C_{\rm H}(\psi_{\theta},W)$. The asymptotic unbiasedness of our estimator follows from $\bE[\htheta-\theta] = O(\norm{\Delta\phi}^2) + O(\prob(T^c)) \rightarrow 0$.

Above, we have proved \eqref{eq:hcrb-cov}. In order to prove \eqref{eq:qcrb-cov}, we notice that it is sufficient to show there exists another asymptotically unbiased estimator such that 
$\forall W \succ 0$, $\lim_{n\rightarrow \infty} n \trace(W V) = 2 C_{\rm F}(\rho_{\theta},W)$. 
Our measurement strategy is exactly the same as previous, except that in Stage 2, we need to switch to another rank-one measurement $M$~\cite{li2016fisher} that achieves $I(\psi_{\ctheta,\cvarphi},M) = \frac{1}{2} J(\psi_{\ctheta,\cvarphi})$. \eqref{eq:sm-1} and \eqref{eq:sm-2} (with $C_{\rm H}(\psi_{\ctheta,\cvarphi},W^*)$ replaced by $2C_{\rm F}(\psi_{\ctheta,\cvarphi},W^*)$) still hold true for the new measurement because they only require the local unbiasedness condition for $\htheta^{(k)}(s,\ell_k)_i$ (and $I(\psi_{\ctheta,\cvarphi},M) = \frac{1}{2} J(\psi_{\ctheta,\cvarphi})$). \eqref{eq:qcrb-cov} is then proved because all other proof steps follow exactly as above. 
\end{proof}

We remark that we only require $\rho_\theta$ to have a non-degenerate positive spectrum instead of an entirely non-degenerate spectrum, a prior requirement for the HCRB attainability previously considered difficult to remove~\cite{yang2019attaining}.

\onecolumngrid
\newpage
\appendix

\setcounter{page}{1}

%\onecolumngrid

\setcounter{theorem}{0}
\setcounter{proposition}{0}
\setcounter{lemma}{0}
\setcounter{figure}{0}
\renewcommand{\thefigure}{S\arabic{figure}}
\renewcommand{\thelemma}{S\arabic{lemma}}
\renewcommand{\thetheorem}{S\arabic{theorem}}
\renewcommand{\thecorollary}{S\arabic{corollary}}
\renewcommand{\theproposition}{S\arabic{proposition}}
\renewcommand{\theHfigure}{Supplement.\arabic{figure}}
\renewcommand{\theHlemma}{Supplement.\arabic{lemma}}
\renewcommand{\theHtheorem}{Supplement.\arabic{theorem}}
\renewcommand{\theHcorollary}{Supplement.\arabic{corollary}}
\renewcommand{\theHequation}{Appendix.\Alph{section}.\arabic{equation}}

\begin{center}
    \textbf{Supplemental Material: Quantum metrology of mixed states via purification}
\end{center}

\section{Precise statements of regularity conditions for the parameterized state}
\label{app:regularity}

The following regularity conditions are assumed throughout this work. 
\begin{enumerate}[(1), wide, labelwidth=!, labelindent=0pt]
    \item $\Theta \subseteq \bR^m$ is the domain of $\theta$ which we assume to be open and bounded. Unless stated otherwise, all assumptions below applies to all $\theta \in \Theta$. 
    \item $\rho_\theta \in \bR^{d \times d}$ is three times continuously differentiable. That means, all partial derivatives of $\rho_\theta$ up to third order, including mixed partial derivatives, exist and are continuous. 
    \item $\rho_\theta$ is rank-$r$, where $r$ is fixed and known in advance. Furthermore, its positive spectrum is non-degenerate, i.e. its positive eigenvalues are distinct. 
    \item $\rho_\theta$ can be written as $\sum_{i=1}^r \lambda_{i,\theta}\ket{e_{i,\theta}}\bra{e_{i,\theta}}$ where $\lambda_{i,\theta} > 0$ and $\ket{e_{i,\theta}}$ is the corresponding eigenstate. We also define $\lambda_{i,\theta} = 0$ when $i \geq r+1$ and $\{\ket{e_{i,\theta}}\}_{i=r+1}^d$ to be an orthonormal set of states in the kernal of $\rho_\theta$. We assume for all $i$, $\lambda_{i,\theta}$ and $\ket{e_{i,\theta}}$ are three times continuously differentiable.
    % We assume $r$ is invariant for $\theta \in \Theta$, and the positive spectrum of $\rho_\theta$ is \emph{non-degenerate} so that the diagonalization is well-defined.  
    \item $\theta$ is identifiable, i.e. $\rho_\theta$ is injective. 
    \item $J(\rho_\theta) \succ 0$. Or equivalently, $\{\partial_{\theta_i}\rho_\theta\}_{i=1}^m$ are linearly independent. 
\end{enumerate}

The following additional assumptions are needed for \thmref{thm:hcrb}. 
\begin{enumerate}[resume*, wide, labelwidth=!, labelindent=0pt]
    \item The minimum non-zero eigenvalue of $\rho_\theta$ is above some positive constant. 
    \item The minimum (maximum) eigenvalue of $J(\rho_\theta)$ is above (below) some positive constant. 
    \item There is another open bounded set $\Theta^\diamond$ such that $\Theta \Subset \Theta^\diamond$, i.e. $\Theta \subseteq \overline{\Theta} \subseteq \Theta^\diamond$ where $\overline{\Theta}$ is the closure of $\Theta$. We assume the definition of $\rho_\theta$ along with all regularity assumptions (except for regularity condition~(1)), as stated in this appendix, extends from $\theta \in \Theta$ to $\theta \in \Theta^\diamond$. 
\end{enumerate}

These regularity conditions are used in the main text and the appendices below.

\section{Positivity of \texorpdfstring{$J(\psi_{\theta,\varphi})$}{QFIM} in the setting of \texorpdfstring{\thmref{thm:purification}}{\thmref{thm:purification}}}
\label{app:positivity}

Here we prove that whenever $J(\rho_{\theta}) \succ 0$, $J(\psi_{\theta,\varphi}) \succ 0$. 
This result is needed in the proof of \thmref{thm:purification}. 

We first note that the positivity of $J(\psi_{\theta,\varphi})$ is equivalent to the following condition: 
\begin{equation}
    \big\{ \partial_{\phi_j} \psi_{\phi} \big\}_{j=1}^{m^*} = \big\{ \partial_{\theta_i} \psi_{\theta,\varphi} \big\}_{i=1}^m \cup \big\{ \partial_{\varphi_j} \psi_{\theta,\varphi} \big\}_{j=1}^{r^2-1} 
\end{equation}
are linearly independent. In particular, if there exists some vector $v \in \bR^{m+r^2-1}$ such that $v^\top J(\psi_{\phi}) v = 0$, then $\sum_i v_i \partial_{\phi_i} \psi_{\phi} = 0$ and vice versa. Thus, we only need to show the linear independence of the partial derivatives in $\Theta \times \Omega$. 

Assume there exists some $v = (a, b) \in \bR^{m} \times \bR^{r^2-1}$ such that 
\begin{equation}
    \sum_{i=1}^{m} a_i \partial_{\theta_i} \psi_{\theta,\varphi} + \sum_{i=1}^{r^2-1} b_i \partial_{\varphi_i} \psi_{\theta,\varphi} = 0,
\end{equation}
where 
\begin{equation}
    \psi_{\theta,\varphi} = \ket{\psi_{\theta,\varphi}}\bra{\psi_{\theta,\varphi}}, \quad \ket{\psi_{\theta,\varphi}} = \sum_{j=1}^r \sqrt{\lambda_{j,\theta}}\ket{e_{j,\theta}}_S \ket{f_{j,\varphi}}_E
\end{equation}
We want to prove $v = 0$. We first take the partial trace over the environment $E$. 
\begin{equation}
    \sum_{i=1}^{m} a_i \partial_{\theta_i} \rho_{\theta} + \sum_{i=1}^{r^2-1} b_i \partial_{\varphi_i} \rho_{\theta} = \sum_{i=1}^{m} a_i \partial_{\theta_i} \rho_{\theta} = 0. 
\end{equation}
Since $J(\rho_\theta) \succ 0$ and $\{\partial_{\theta_i} \rho_{\theta}\}_{i}$ are linearly independent, we have $a_i = 0$. Thus 
\begin{equation}
 \sum_{i=1}^{r^2-1} b_i \partial_{\varphi_i} \psi_{\theta,\varphi} = 0. 
\end{equation}
Since $\partial_{\varphi_i} \ket{f_{j,\varphi}} = -i H_{i,\varphi} \ket{f_{j,\varphi}}$ for all $j$,
\begin{equation}
 \sum_{i=1}^{r^2-1} b_i \partial_{\varphi_i} \psi_{\theta,\varphi} = 
 -i \Big[ \id_S \otimes \sum_{i=1}^{r^2-1} b_i  H_{i,\varphi}, \psi_{\theta,\varphi} \Big] = 0, \quad \Rightarrow \quad \Big( \id_S \otimes \sum_{i=1}^{r^2-1} b_i  H_{i,\varphi} \Big) \ket{\psi_{\theta,\varphi}} = c \times \ket{\psi_{\theta,\varphi}}, 
\end{equation}
for some $c \in \bR$. 
Since $\lambda_{j,\theta} \neq 0$, we must have 
\begin{equation}
    \sum_{i=1}^{r^2-1} b_i  H_{i,\varphi} \ket{f_{j,\varphi}} = c \ket{f_{j,\varphi}},\quad \forall j. 
\end{equation}
Given that $\ket{f_{j,\varphi}}$ forms a complete basis of $E$, we must have $\sum_{i=1}^{r^2-1} b_i  H_{i,\varphi} = 0$. Then $b_i = 0$ because of the linear independence of $\{H_{i,\varphi}\}$. Thus, we have shown $(a,b) = 0$ when $\sum_{i=1}^{m} a_i \partial_{\theta_i} \psi_{\theta,\varphi} + \sum_{i=1}^{r^2-1} b_i \partial_{\varphi_i} \psi_{\theta,\varphi} = 0$. The linear independence of $\partial_{\phi} \psi_{\phi}$ thus implies the positivity of $J(\rho_\phi)$. 

\section{Justification of the parametrization step in the proof of  \texorpdfstring{\thmref{thm:hcrb}}{Theorem~2}}

In Stage 1 of the estimation protocol in \thmref{thm:hcrb}, conditioned on the event $T = T_{U,s,n_1}$ occurs, we obtain a relatively good coarse estimate of $(\ctheta,\cU)$. We then parameterize the quantum state as $\psi_{\theta,\varphi} := \psi_{\theta,U}$, taking $\psi_{\ctheta,\cvarphi} = \psi_{\ctheta,\cU}$. Stage 2 performs optimal local parameter estimation on the remaining states to obtain the final precise estimate of $\theta$. In order to show the attainability of the HCRB (or twice the QCRB), we need to show $(\theta,\varphi)$ and $(\ctheta,\cvarphi)$ are sufficiently close (as $n \rightarrow \infty$) such that the locally unbiased estimator at $(\ctheta,\cvarphi)$ performs well when the true values are $(\theta,\varphi)$. Informally speaking, the goal of this appendix is to prove 
\begin{center}
    \emph{When $\psi_{\theta,U}$ and $\psi_{\ctheta,\cU}$ are sufficiently close, $(\theta,\varphi)$ and $(\ctheta,\cvarphi)$ are also (uniformly) close.}
\end{center}
Here ``\emph{uniformly}'' means the closeness between $(\theta,\varphi)$ and $(\ctheta,\cvarphi)$ is independent of the specific values of $\theta,U,\ctheta,\cU$, instead it depends purely on the closeness between $\psi_{\theta,U}$ and $\psi_{\ctheta,\cU}$. This requirement is the key reason why additional regularity conditions (7) and (8) are introduced for the proof of \thmref{thm:hcrb}.
To show this, we aim to prove several properties of the parameterization step. We provide an overview here, followed by rigorous mathematical proofs below. 
\begin{enumerate}[1., wide, labelwidth=!, labelindent=0pt]
    \item \textbf{Identifiability.} We aim to show that when $\varphi$ is within a sufficiently small region, there is only one possible $\varphi$ such that $\psi_{\theta,U}$ can be represented as $\psi_{\theta,\varphi}$. 
    \item \textbf{Expressibility.} We aim to show also when $\psi_{\theta,U}$ and $\psi_{\ctheta,\cU}$ are sufficiently close, there always exists some $\varphi$ within a small region such that $\psi_{\theta,\varphi}$ can represent $\psi_{\theta,U}$.
    \item \textbf{Inverse stability.} Using the above two properties, we manage to narrow down the domain of $(\theta,\varphi)$ to a small region. Then we can solve the inverse stability problem using the technique of Tayler expansion. Within a small region $\varphi$, the closeness between states can then be translated to the closeness between parameters. 
\end{enumerate}
Below we use the following definitions from the main text. 
\begin{gather}
\label{eq:psi}
    \ket{\psi_{\theta,\varphi}} := \sum_{j=1}^r  \sqrt{\lambda_{j,\theta}}\ket{e_{j,\theta}}_S \exp\left(-i\sum_{i=1}^{r^2-1} \varphi_i H_i\right)\cU\ket{j}_E,
\\
    \ket{\psi_{\theta,U}} := \sum_{j=1}^r  \sqrt{\lambda_{j,\theta}}\ket{e_{j,\theta}}_S U\ket{j}_E.
\end{gather}
Note that $\ket{\psi_{\theta,\varphi}}$ implicitly depends on $\cU$. We also assume 
\[
\cvarphi \equiv 0
\]
throughout. 
In this appendix, we assume implicitly $\cU$, $U$ belong to $\text{PU}(r)=\text{U}(r)/\text{U}(1)$. $\dpu(U,U_0) = \min_{\alpha}\norm{U-U_0 e^{i\alpha}}$ is a metric on $\text{PU}(r)$. We also view quantum states as vectors in $\bC\mathbb{P}^{dr-1}$ which is $\bC^{dr}$ with nonzero complex scalars quotiented out. Note that $\ket{\psi_{\theta,U}}$ is three times continuously differentiable with respect to $\theta$, and smooth with respect to $U$.

\subsection{Identifiability and expressibility}

We first prove the following two lemmas and use them to prove \lemmaref{lemma:ident}. 
\begin{lemma}
\label{lemma:2}
Given ${{\kappa_1}} > 0$, there exists ${{\eta}}> 0$ such that when $1 - |\!\braket{\psi_{\theta,U}|\psi_{\ztheta,\zU}}\!|^2 < \eta$, $\norm{\theta - \ztheta} < {{\kappa_1}}$ and $\dpu(U,\zU) < {{\kappa_1}}$. Here $(\theta,U), (\theta^0,U^0) \in \overline{\Theta} \times {\rm PU}(r)$. 
\end{lemma}
\begin{proof}
Consider the set 
\begin{equation}
    K_{{\kappa_1}} = \big\{(\theta,U,\ztheta,\zU) \in \overline{\Theta} \times {\rm PU}(r) \times \overline{\Theta} \times {\rm PU}(r) \big| \norm{\theta - \ztheta} \geq {{\kappa_1}} \text{~or~} \dpu(U,\zU) \geq {{\kappa_1}}\big\}. 
\end{equation}
Since $\overline{\Theta}$ and ${\rm PU}(r)$ are compact, $K_{{\kappa_1}}$ is compact. Consider
\begin{equation}
    g(\theta,U,\ztheta,\zU) = 1 - |\!\braket{\psi_{\theta,U}|\psi_{\ztheta,\zU}}\!|^2. 
\end{equation}
It is a continuous function on the compact set $\overline{\Theta} \times {\rm PU}(r) \times \overline{\Theta} \times {\rm PU}(r)$. Moreover, in $K_{{\kappa_1}}$, we have $\theta \neq \theta^0$ or $U \neq U^0$ and $g(\theta,U,\ztheta,\zU)$ is positive. To see this, assume we have $g(\theta,U,\ztheta,\zU) = 0$. By identifiability of $\theta$ and monotonicity of fidelity (when tracing out $E$), we must have $\theta = \ztheta$. Since $\lambda_{j,\theta}$ are distinct for $j$ (due to the non-degenerate positive spectrum), we then have $U = \zU$.  
Therefore, there is a positive minimum of $g(\theta,U,\ztheta,\zU)$ in $K_\kappa$. We choose it to be $\eta$.  
\end{proof}
\begin{lemma}
\label{lemma:4}
(1)~There exists some constant $\kappa_4 > 0$ such that $\exp\Big(-i\sum_{i=1}^{r^2-1} \varphi_i H_i\Big)  \in {\rm PU}(r)$ is injective when $\norm{\varphi} \leq \kappa_4$. (2)~For any $\kappa_3 > 0$, there exists some $\kappa_2 > 0$ such that $\exp\Big(-i\sum_{i=1}^{r^2-1} \varphi_i H_i\Big)$ can represent any unitary $U \in {\rm PU}(r)$ that satisfies $\dpu\Big(U,\id\Big) < \kappa_2$ using some $\varphi$ satisfying $\norm{\varphi} < \kappa_3$. 
\end{lemma}
\begin{proof}
(1) By the inverse function theorem~\cite{LeeSmoothManifolds} and the fact that $\{H_i\}$ are traceless and linearly independent, the map $\varphi \mapsto \exp\left(-i\sum_{i=1}^{r^2-1}\varphi_i H_i\right)$ is a local diffeomorphism at $\varphi=0$. This implies the existence of $\kappa_4$ such that when $\norm{\varphi}$ is below it, the map becomes injective. 
(2) For any $\kappa_3>0$, the image of
$\{\varphi:\|\varphi\|<\kappa_3\}$ under the map contains an open
neighborhood of $\id$ in $\mathrm{PU}(r)$.
Since $d_{\mathrm{PU}}$ induces the standard topology on
$\mathrm{PU}(r)$, this neighborhood contains a
$d_{\mathrm{PU}}$-ball centered at $\id$ with some
radius $\kappa_2>0$, proving the claim.
\end{proof}

\begin{lemma}[Identifiability and expressibility]\label{lemma:ident}
    $(\theta,U), (\ctheta,\cU) \in \overline{\Theta} \times {\rm PU}(r)$. For any fixed $\kappa \in (0,\kappa_4)$, we can choose $\eta > 0$ such that when $1 - |\!\braket{\psi_{\theta,U}|\psi_{\ctheta,\cU}}\!|^2 < \eta$, there exists some unique $\varphi$ such that $(\theta,\varphi) \in B_{\kappa}(\cphi):=\{\phi|\|\phi-\cphi\| < \kappa\}$ and $\psi_{\theta,U} = \psi_{\theta,\varphi}$. 
\end{lemma}
\begin{proof}
Choose $\kappa_3 \leq \kappa_{4}$ such that 
\begin{equation}
\textstyle\label{eq:subset-1}
        \{ (\theta,\varphi) | \norm{\theta-\ctheta} < \kappa_3, \norm{\varphi-\cvarphi} = \norm{\varphi} < \kappa_3 \} \subseteq B_{\kappa}(\cphi). 
\end{equation}
Choose $0 < \kappa_2 < \kappa_3$ using \lemmaref{lemma:4} such that $\exp\Big(-i\sum_{i=1}^{r^2-1} \varphi_i H_i\Big)$ can represent any unitary $U \in {\rm PU}(r)$ that satisfies $\dpu\Big(U,\id\Big) < \kappa_2$. 
From \lemmaref{lemma:2}, there exists $\eta > 0$ such that 
\begin{equation}
\textstyle
    \norm{\theta-\ctheta} < \kappa_2 \leq \kappa_3,\quad \dpu\Big( U \cU^\dagger,\id\Big) < \kappa_2. 
\end{equation}
According to \lemmaref{lemma:4}, there exists some $\norm{\varphi} \leq \kappa_3$ such that 
\begin{equation}
\textstyle
    U \cU^\dagger =  \exp\left(-i\sum_{i=1}^{r^2-1} \varphi_i H_i\right) . 
\end{equation}
Thus $(\theta,\varphi) \in B_{\kappa}(\cphi)$ and $\eta$ satisfies our requirement. 
\end{proof}

\subsection{Inverse stability}

Set $\varphi^0 = 0$ throughout this section. 
Consider the Taylor expansion in $\Delta\phi := \phi - \zphi$. 
\begin{gather}
\label{eq:taylor}
    \big|\!\braket{\psi_{\theta,\varphi}|\psi_{\ztheta,\zvarphi}}\!\big|^2 = 1 - \frac{1}{4}(\Delta\phi)^\top J(\psi_{\ztheta,\zvarphi}) \Delta\phi + R_3(\Delta\phi)
    ,
\\
J(\psi_{\ztheta,\zvarphi}) = \begin{pmatrix}
    J(\psi_{\ztheta,\zvarphi})_{SS} & J(\psi_{\ztheta,\zvarphi})_{SE} \\
    J(\psi_{\ztheta,\zvarphi})_{ES} & J(\psi_{\ztheta,\zvarphi})_{EE}
\end{pmatrix}, \quad \text{$S$ and $E$ correspond to $\theta$ and $\varphi$ respectively,}
\\
(J(\psi_{\ztheta,\zvarphi})_{EE})_{ij} = 4 \left(\frac{1}{2}\braket{\{H_i,H_j\}} - \braket{ H_i}\braket{H_j}\right) \bigg|_{{\Tr_S(\psi_{\ztheta,\zvarphi})}}.\label{eq:HH}
\end{gather}
$\braket{\cdot}$ is evaluated at $\Tr_S(\psi_{\ztheta,\zvarphi})$ which is $\psi_{\ztheta,\zvarphi}$ with the probe system $S$ traced out.

The last term $R_3(\Delta\phi)$ in \eqref{eq:taylor} is the remainder term in the Taylor expansion which is a degree-3 polynomial of $\Delta\phi$ where the coefficients are described by the first, second and third derivatives of $\psi_{\phi}$. Given that $\psi_{\phi}$ is three times continuously differentiable, when $\phi$ is within a compact set, the coefficients are uniformly bounded. We can safely write the remainder term as 
\begin{equation}
\label{eq:remainder}
    R_3(\Delta\phi) = O(\norm{\Delta\phi}^3), \text{~~when~} (\theta,\varphi) \in \overline{\Theta} \times \{\varphi| \norm{\varphi} \leq \kappa_4 \}, 
\end{equation}
which uniformly decays as $\Delta\phi \rightarrow 0$.

We have shown in \appref{app:positivity} that $J(\psi_{\ztheta,\zvarphi}) \succ 0$ and $J(\psi_{\ztheta,\zvarphi})_{EE} \succ 0$ because  $J(\rho_{\ztheta}) \succ 0$ as implied by regularity condition~(6). Here we aim to prove more stringent constraints on $J(\psi_{\ztheta,\zvarphi})$ by bounding its largest and smallest eigenvalues. 

Take $v \in \bR^{r^2-1}$ and $H_v = \sum_i v_i H_i$, then 
\begin{equation}
\label{eq:Hv}
v^\top J(\psi_{\ztheta,\zvarphi})_{EE} v = 4 \Bigr( \frac{1}{2}\braket{\{H_v,H_v\}} - \braket{ H_v}\braket{H_v} \Bigl) = 4 ( \braket{H_v^2} - \braket{H_v}^2). 
\end{equation}
From regularity condition~(7), there is some positive constant $\epsilon_0$ such that 
\begin{equation}
    \rho_\ztheta \succeq \epsilon_0 \Pi_\ztheta \quad \Rightarrow\quad  
    \Tr_S(\psi_{\ztheta,\zvarphi}) \succeq \epsilon_0 \id  \quad \Rightarrow\quad  
    J(\psi_{\ztheta,\zvarphi})_{EE} \succeq \epsilon_1 \id, 
\end{equation}
where $\Pi_\ztheta$ is the projection onto the support of $\rho_\theta$, $\epsilon_1$ is some positive constant depending on the choice of $\{H_i\}_i$ found by minimizing \eqref{eq:Hv} over all $v$'s with $\norm{v} = 1$ and states with minimum eigenvalues no smaller than $\epsilon_0$. 

We also see that $J(\psi_{\ztheta,\zvarphi})$ is congruent to 
\begin{equation}
\begin{pmatrix}
   ((J(\psi_{\ztheta,\zvarphi})^{-1})_{SS})^{-1} & 0 \\ 
    0 & J(\psi_{\ztheta,\zvarphi})_{EE} \\
\end{pmatrix}
= \begin{pmatrix}
    J(\rho_{\ztheta}) & 0 \\ 
    0 & J(\psi_{\ztheta,\zvarphi})_{EE} \\
\end{pmatrix},
\end{equation}
under the following transformation: 
\begin{multline}
\label{eq:conrgu}
    \begin{pmatrix}
    J(\psi_{\ztheta,\zvarphi})_{SS} & J(\psi_{\ztheta,\zvarphi})_{SE} \\
    J(\psi_{\ztheta,\zvarphi})_{ES} & J(\psi_{\ztheta,\zvarphi})_{EE}
\end{pmatrix}
= \\
\begin{pmatrix}
    \id  &  J(\psi_{\ztheta,\zvarphi})_{SE} J(\psi_{\ztheta,\zvarphi})_{EE}^{-1}\\ 
    0 & \id 
\end{pmatrix}
\begin{pmatrix}
    J(\rho_{\ztheta}) & 0 \\ 
    0 & J(\psi_{\ztheta,\zvarphi})_{EE} \\
\end{pmatrix}
\begin{pmatrix}
    \id  & 0 \\ 
    J(\psi_{\ztheta,\zvarphi})_{EE}^{-1} J(\psi_{\ztheta,\zvarphi})_{ES} & \id 
\end{pmatrix}. 
\end{multline}
Note that the equations above \eqref{eq:taylor}--\eqref{eq:conrgu} apply to any $\ztheta \in \Theta $. 

\begin{lemma}
\label{lemma:1}
$\exists \epsilon > 0$, s.t.,  $
    J(\psi_{\theta,\varphi})|_{\varphi = 0} \succeq \epsilon \id, \quad \forall \cU$, $\forall \theta \in \overline{\Theta} $. 
\end{lemma}
\begin{proof}
According to regularity condition~(9), we can safely extend $\Theta$ to $\overline{\Theta}$ in the discussion here. Thus 
we only need to find a positive lower bound that is continuous with respect to $\cU$ and $\theta$, since they belongs to a compact set ${\rm PU}(r) \times \overline{\Theta}$.
From \eqref{eq:conrgu} and regularity condition~(8) that $J(\rho_{\theta}) \succeq \epsilon_2 \id$ for some constant $\epsilon_2$, 
\begin{align}
    \norm{J(\psi_{\theta,0})^{-1}} &\leq \norm{\begin{pmatrix}
    \id  & 0 \\ 
    J(\psi_{\theta,0})_{EE}^{-1} J(\psi_{\theta,0})_{ES} & \id 
\end{pmatrix}}^2 \cdot \max\bigg\{\norm{J(\rho_{\theta})^{-1}},\norm{(J(\psi_{\theta,0})_{EE})^{-1}}\bigg\}\\
&\leq \left(\norm{J(\psi_{\theta,0})_{EE}^{-1} J(\psi_{\theta,0})_{ES}} + 2\right)^2 \cdot \max\{1/\epsilon_1,1/\epsilon_2\}\\
&\leq \left(1/\epsilon_1 \norm{J(\psi_{\theta,0})_{ES}} + 2\right)^2 \cdot \max\{1/\epsilon_1,1/\epsilon_2\}. 
\end{align}
Thus it is sufficient to show $J(\psi_{\theta,0})_{ES}$ is continuous in $(\theta,\cU)$. 
\begin{equation}
    J(\psi_{\theta,0})_{ij} = 4\Re[\braket{\partial_i\psi|\partial_j\psi} - \braket{\partial_i\psi|\psi}\braket{\psi|\partial_j\psi} ] |_{\varphi = 0}.
\end{equation}
From \eqref{eq:psi}, all terms are continuous in $\theta$ and $\cU$. Thus we've obtained a continuous function as an upper bound on $\norm{J(\psi_{\theta,0})^{-1}}$.    
\end{proof}

Note that both $\norm{J(\rho_{\theta})}$ and $\norm{(J(\psi_{\theta,\varphi})_{EE})}$ are also bounded from above, using regularity condition~(8), and \eqref{eq:HH} (viewing $H_i$ as pre-determined constant matrices), respectively. Thus, the above discussion can also be applied to prove similarly the following lemma (although we will not use it until \appref{app:boundedness}).  
\begin{lemma}
\label{lemma:3}
$\exists L > 0$, s.t.,  $
    J(\psi_{\theta,\varphi})|_{\varphi = 0} \preceq L \cdot \id, \quad \forall \cU$, $\forall \theta \in \overline{\Theta} $.
\end{lemma}

Now we are ready to state our final lemma that combines \lemmaref{lemma:ident} and the inverse stability result, ready for use in the main text. 
\begin{lemma}[Inverse stability]
    There exists some $\kappa > 0$ and $\eta > 0$ such that for any $\eta' < \eta$, when $1 - |\!\braket{\psi_{\theta,U}|\psi_{\ctheta,\cU}}\!|^2 < \eta'$, 
    \begin{enumerate}
        \item There exists some unique $\varphi$ such that $(\theta,\varphi) \in B_{\kappa}(\cphi):=\{\phi|\|\phi-\cphi\| < \kappa\}$ and $\psi_{\theta,U} = \psi_{\theta,\varphi}$. 
        \item $\norm{\phi - \cphi} = O(\sqrt{\eta'})$.  
    \end{enumerate}
\end{lemma}

\begin{proof}
Using \lemmaref{lemma:1} and \eqref{eq:remainder}, there exists some constant $\kappa \in (0,\kappa_4)$ such that when $\phi \in B_{\kappa}(\cphi)$, 
\begin{equation}
    \frac{\abs{ R_3(\phi-\cphi)} }{(\phi-\cphi)^\top J(\psi_{\ctheta,\cvarphi}) (\phi-\cphi)} \leq \frac{\abs{ R_3(\phi-\cphi)}}{\epsilon \norm{\phi-\cphi}^2} \leq \frac{1}{10}.
\end{equation}
Thus, using \eqref{eq:taylor}, 
\begin{equation}
    1 - \big|\!\braket{\psi_{\theta,\varphi}|\psi_{\ctheta,\cvarphi}}\!\big|^2 \geq \frac{1}{10}(\phi-\cphi)^\top J(\psi_{\ctheta,\cvarphi}) (\phi-\cphi) \geq \frac{\epsilon}{10}(\phi-\cphi)^\top (\phi-\cphi) = \frac{\epsilon}{10} \norm{\phi-\cphi}^2. 
\end{equation}
Fixing $\kappa$, we can pick $\eta$ using \lemmaref{lemma:ident}. The first condition is automatically satisfied, and the second also holds using the above inequality.  
\end{proof}

\section{Boundedness of locally optimal estimators in the proof of \texorpdfstring{\thmref{thm:hcrb}}{Theorem~2}}
\label{app:boundedness}

Here we briefly review the optimal measurement and the corresponding locally unbiased estimator that achieves the HCRB (or twice the QCRB). We also show the boundedness of the $\theta$-estimators. Note that here by \emph{boundedness} we mean the locally unbiased estimators are bounded by a universal constant that does not depend on $(\ctheta,\cU)$. The important property we use below are \lemmaref{lemma:1} and \lemmaref{lemma:3}. 

We first state Matsumoto's HCRB for pure states. It was shown in \cite{matsumoto2002new} when $\tW \succ 0$,
\begin{equation}
C_{\rm H}(\psi_\phi,\tW) = \min_{\vX=\{\ket{x_i}\}_{i=1}^{m^*},\ket{x_i}\in \mH \oplus \bC^{m^*}} \trace(\tW Z(\vX)),
\end{equation}
where $    Z(\vX)_{ij} = \braket{x_i|x_j}$  and $\ket{x_i}$ satisfies
\begin{equation}
\braket{x_i|\psi_\phi} = 0,\quad 2\Re[\braket{x_i|\partial_j\psi_\phi}] = \delta_{ij},\quad \Im[Z(\vX)] = 0. 
\end{equation}
Here $\mH$ is the Hilbert space $\psi_\phi$ acting on, and $\mH \oplus \bC^{m^*}$ means an enlarged Hilbert space obtained by adjoining an $m^*$-dimensional auxiliary space. We also abuse the notation and think of $\ket{\psi_\phi}$ and $\ket{\partial_j\psi_\phi}$ as vectors in $\mH \oplus \bC^{m^*}$ which are supported only on subspace $\mH$.

In the situation we consider, $\mH = \mH_{SE}$, $\dim(\mH) = dr$ and $\widetilde{W} = W^*$ is a singular, block-diagonal matrix whose lower right corner is zero. $\phi = (\theta,\varphi)$ where $\theta$ are the parameters of interest and $\varphi$ are the nuisance parameters. We use the following modified form of the Mastumoto bound: 
\begin{equation}
\label{eq:matsumoto-1}
C_{\rm H}(\psi_\phi,W^*) = \min_{\vX=\{\ket{x_i}\}_{i=1}^{m},\ket{x_i}\in \mH \oplus \bC^{m}} \trace(W Z(\vX)),
\end{equation}
where $    Z(\vX)_{ij} = \braket{x_i|x_j}$  and $\ket{x_i}$ satisfies for any $i,j \in [m]$ and $a \in [r^2-1]$, 
\begin{equation}
\braket{x_i|\psi_\phi} = 0,\quad 2\Re[\braket{x_i|\partial_{\theta_j}\psi_\phi}] = \delta_{ij},\quad 2\Re[\braket{x_i|\partial_{\varphi_a}\psi_\phi}] = 0,\quad \Im[Z(\vX)] = 0. 
\end{equation}
Consider the enlarged Hilbert space $\mH \oplus \bC^{m}$. Given the optimal $\{\ket{x_i}\} \in \mH \oplus \bC^{m}$, the following projective measurement and locally unbiased estimator is optimal in $\mH \oplus \bC^{m}$: 
    \begin{equation}
    M_\bell = \ket{e_\bell}\bra{e_\bell},\quad \hat\theta_i(\bell) = \frac{\braket{e_\bell|x_i}}{\braket{e_\bell|\psi_\phi}} + \theta_i,
    \end{equation}
    where $\{\ket{e_\bell}\}$ can be any orthonormal basis of $\mH \oplus \bC^{m}$ such that for all $\bell$, the inner product between $\ket{e_\bell}$ and vectors in $\{\ket{\psi_\phi},\ket{x_i},\forall i\}$ is real. We can therefore, assume without loss of generality, 
    \begin{equation}
        \forall \bell \in [\dim(\mH) + m], \quad \braket{e_\bell|\psi_\phi} = \frac{1}{\sqrt{\dim(\mH) + m}}, 
    \end{equation}
    making the estimator \emph{bounded}, because
    \begin{multline}
        \abs{\braket{e_\bell|x_i}}^2 \leq \braket{x_i|x_i} \leq \norm{W^{-1}}\Tr( W V) \leq 2\norm{W^{-1}}\Tr( W^* (J(\psi_{\theta,\varphi})^{-1})) \\= 2\norm{W^{-1}}\trace(W (J(\psi_{\theta,\varphi})^{-1})_{SS}) = 2\norm{W^{-1}}\trace(W J(\rho_{\theta})^{-1}), 
    \end{multline}
    which is bounded due to regularity condition~(8) and the fact that $W$ is a known, positive matrix. 
    To see the estimator is locally unbiased, 
    \begin{gather}
    \bE_\bell [\hat\theta_i(\bell)] = \theta_i + \sum_\bell  \frac{\braket{e_\bell|x_i}}{\braket{e_\bell|\psi_\phi}}  \braket{\psi_\phi | M_\bell | \psi_\phi} = \theta_i,
    \\    
    \partial_{\theta_j} \bE_\bell [\hat\theta_i(\bell)] = \sum_\bell  \frac{\braket{e_\bell|x_i}}{\braket{e_\bell|\psi_\phi}}  \partial_{\theta_j}\braket{\psi_\phi | M_\bell | \psi_\phi}
    = \sum_\bell  \frac{\braket{e_\bell|x_i}}{\braket{e_\bell|\psi_\phi}}  2 \Re[\braket{e_\bell | \partial_{\theta_j} \psi_\phi}\braket{\psi_\phi|e_\bell}] = \delta_{ij},\\
    \partial_{\varphi_a} \bE_\bell [\hat\theta_i(\bell)] = \sum_\bell  \frac{\braket{e_\bell|x_i}}{\braket{e_\bell|\psi_\phi}}  \partial_{\varphi_a}\braket{\psi_\phi | M_\bell | \psi_\phi}
    = \sum_\bell  \frac{\braket{e_\bell|x_i}}{\braket{e_\bell|\psi_\phi}}  2 \Re[\braket{e_\bell | \partial_{\varphi_a} \psi_\phi}\braket{\psi_\phi|e_\bell}] = 0.
    \end{gather}
    They are also optimal because 
    \begin{equation}
    \begin{split}
    V(\psi_\phi,M,\hat\theta)_{ij} 
    &= \sum_\bell (\hat\theta_i(\bell)-\theta_i)(\hat\theta_j(\bell)-\theta_j) \braket{\psi_\phi|M_\bell|\psi_\phi} \\
    &= \sum_\bell \frac{\braket{x_i | e_\bell}}{\braket{\psi_\phi | e_\bell }}  \frac{\braket{e_\bell|x_j}}{\braket{e_\bell|\psi_\phi}}  \braket{\psi_\phi | M_\bell | \psi_\phi} = \sum_\bell {\braket{x_i | e_\bell}} {\braket{e_\bell|x_j}} = \braket{x_i|x_j}. 
    \end{split}
    \end{equation}
Note that $M$ is a projective measurement on $\mH \oplus \bC^{m}$. We can restrict it to a rank-one (but not necessarily projective) measurement on $\mH$ simply by projecting it to $\mH$, because $\braket{\psi_\phi|\Pi_\mH M_\bell \Pi_\mH|\psi_\phi} = \braket{\psi_\phi|M_\bell|\psi_\phi}$ and then 
\begin{equation}
    V(\psi_\phi,\Pi_\mH M \Pi_\mH,\hat\theta)_{ij} = \sum_\bell (\hat\theta_i(\bell)-\theta_i)(\hat\theta_j(\bell)-\theta_j) \braket{\psi_\phi|\Pi_\mH M_\bell \Pi_\mH|\psi_\phi} = \braket{x_i|x_j}. 
\end{equation}
The optimal measurement $\{\ket{b_{\bell}}\bra{b_{\bell}}\}$ we use in the main text is 
\begin{equation}
    \ket{b_{\bell}} = ( \Pi_\mH \ket{e_{\bell}} )_{\mH},
\end{equation}
where we use $_\mH$ to represent a state in $\mH$, and $\ket{b_{\bell}}$ is unnormalized. 

Next, we briefly review Li \emph{et al.}'s results on Fisher-symmetric measurements~\cite{li2016fisher}. A measurement is called Fisher-symmetric measurement for pure states, if and only if 
\begin{equation}
    I(\psi_\phi,M) = \frac{1}{2} J(\psi_\phi). 
\end{equation}
The coefficient $1/2$ is optimal if we want the above condition to hold for all pure states. However, we do allow $M$ to be chosen based on the knowledge of $\psi_\phi$. \cite{li2016fisher} shows there exists a $2\dim(\mH)-1$-outcome rank-one measurement $\{\ket{a_\bell}\bra{a_\bell}\}_{\bell \in [2\dim(\mH)-1]}$ (with $\ket{a_\bell}$ unnormalized, $\braket{a_\bell|a_\bell} = \dim(\mH) / (2\dim(\mH)-1)$) that is Fisher-symmetric and
\begin{equation}
   p_{\bell,\phi} = |\!\braket{\psi_\phi | a_\bell}\!|^2 = \frac{1}{2\dim(\mH)-1},
\end{equation}
when $\phi$ is the true value. 
A canonical choice of locally unbiased estimator from classical probability theory that achieves the CFIM at $\phi = \cphi$ is 
\begin{equation}
    \hat\phi(\bell)_i =  \cphi_i + \sum_{j}  (I(\psi_{\cphi},M)^{-1})_{ij} \frac{\partial_j p_{\bell,\phi}}{p_{\bell,\phi}} \bigg|_{\phi = \cphi}. 
\end{equation}
$\hat\theta(\bell)_i$ must also be \emph{bounded} because $p_{\bell,\cphi} = \frac{1}{2\dim(\mH)-1}$, $I(\psi_{\cphi},M) = \frac{1}{2} J(\psi_\cphi)$ is lower bounded by some positive constant (from \lemmaref{lemma:1}), and  
\begin{align}
    \big|\partial_j p_{\bell,\phi}\big| \big|_{\phi = \cphi} &= 2 \big|\Re[\braket{a_\bell|\partial_j\psi_\phi}\!\braket{\psi_\phi|a_\bell}]\big| 
    \big|_{\phi = \cphi}  \leq \norm{ \ket{\partial_j\psi_\phi}\!\bra{\psi_\phi} + \ket{\psi_\phi}\!\bra{\partial_j\psi_\phi}} \big|_{\phi = \cphi}  \\  
    &\leq \sqrt{J(\psi_{\ctheta,\cvarphi})_{jj}} \leq \sqrt{\norm{J(\psi_{\ctheta,\cvarphi})}},
\end{align}
which is bounded using \lemmaref{lemma:3}.

\section{Proof of \texorpdfstring{\eqref{eq:sm-1}}{Eq.24} and \texorpdfstring{\eqref{eq:sm-2}}{Eq.25}}

Here we prove \eqref{eq:sm-1} and \eqref{eq:sm-2} in the main text. They follow directly from the local unbiasedness of $\htheta^{(k)}(s,\ell_k)_i$, i.e. 
\begin{gather}
\bE_{\ell_k} [\htheta^{(k)}(s,\ell_k)] \big|_{\psi_{\ctheta,\cvarphi}} \!= \ctheta, \\
\frac{\partial}{\partial \phi_j}\bE_{\ell_k} [\htheta^{(k)}(s,\ell_k)_i] \big|_{\psi_{\ctheta,\cvarphi}} \!= \delta_{ij},  
\end{gather} 
and 
\begin{equation}
\label{eq:dev}
    \psi_{\phi} = \psi_{\cphi} + (\nabla \psi_{\cphi}) \cdot \Delta\phi + O(\|\Delta\phi\|^2). 
\end{equation}
Note that here the term $O(\|\Delta\phi\|^2)$ converges uniformly (independent of $\theta,\varphi,U,\cU,s$). Recall that 
\begin{equation}
    \ket{\psi_\phi} = \ket{\psi_{\theta,\varphi}} = \sum_{j=1}^r  \sqrt{\lambda_{j,\theta}}\ket{e_{j,\theta}}_S \exp\left(-i\sum_{i=1}^{r^2-1} \varphi_i H_i\right)\cU\ket{j}_E
\end{equation}
The uniform convergence is guaranteed due to the continuity of $\nabla^2 \psi_{\cphi}$ with respect to $(\theta,\cU)$. The compactness of $\overline{\Theta} \times {\rm PU}(r)$ then implies the uniform continuity of $\nabla^2 \psi_{\cphi}$. By the same argument, we also note $\|\nabla \psi_{\cphi}\|$ is upper bounded by a universal constant. Below we also assume boundedness of estimators, as shown from \appref{app:boundedness}.

First, we have 
\begin{equation}
\begin{split}
    \bE_{\ell_k} [\htheta^{(k)}(s,\ell_k)_i] 
    &= 
    \sum_{\ell_k} \htheta^{(k)}(s,\ell_k)_i \bra{b_{\ell_k}} \psi_{\phi} \ket{b_{\ell_k}}\\
    &= \sum_{\ell_k} \htheta^{(k)}(s,\ell_k)_i \bra{b_{\ell_k}}  \big(\psi_{\cphi} + (\nabla \psi_{\cphi}) \cdot \Delta\phi + O(\|\Delta\phi\|^2) \big) \ket{b_{\ell_k}}\\
    &= \ctheta_i + \Delta\theta_i + O(\|\Delta\phi\|^2) = \theta_i + O(\|\Delta\phi\|^2). 
\end{split}
\end{equation}
The overall estimator is $\htheta(s,\ell)_i = \frac{1}{n_2} \sum_{k=1}^{n_2} \htheta^{(k)}(s,\ell_k)_i$. 
We now calculate the variance: 
\begin{equation}
\begin{split}
    ( V_{U,s})_{ij} 
    &= 
\sum_{\ell} (\htheta(s,\ell)_i-\theta_i) (\htheta(s,\ell)_j-\theta_j)  \!\bra{b_\ell}\!\psi_{\theta,\varphi}^{\otimes n_2}\!\ket{b_\ell}\\
&= \sum_{\ell} \left(\frac{1}{n_2} \sum_{k=1}^{n_2} \htheta^{(k)}(s,\ell_k)_i-\theta_i\right) \left(\frac{1}{n_2} \sum_{k=1}^{n_2} \htheta^{(k)}(s,\ell_k)_j-\theta_j\right)  \!\bra{b_\ell}\!\psi_{\theta,\varphi}^{\otimes n_2}\!\ket{b_\ell}
\end{split}
\end{equation}
First, 
\begin{equation}
\label{eq:c7}
\begin{split}
    &\quad\,\sum_{\ell} \left(\htheta^{(1)}(s,\ell_1)_i-\theta_i\right) \left( \htheta^{(1)}(s,\ell_1)_j-\theta_j\right)  \!\bra{b_\ell}\!\psi_{\theta,\varphi}^{\otimes n_2}\!\ket{b_\ell}\\
    &= \sum_{\ell_1} \left(\htheta^{(1)}(s,\ell_1)_i-\theta_i\right) \left( \htheta^{(1)}(s,\ell_1)_j-\theta_j\right)  \!\bra{b_{\ell_1}}\!\psi_{\theta,\varphi}\!\ket{b_{\ell_1}}\\
    &= \sum_{\ell_1} \left(\htheta^{(1)}(s,\ell_1)_i-\theta_i\right) \left( \htheta^{(1)}(s,\ell_1)_j-\theta_j\right)  \!\bra{b_{\ell_1}}\!\psi_{\ctheta,\cU}\!\ket{b_{\ell_1}} + O(\norm{\Delta\phi}).  
    \end{split}
\end{equation}
Second,
\begin{equation}
\label{eq:c8}
\begin{split}
    &\quad\,\sum_{\ell} \left(\htheta^{(1)}(s,\ell_1)_i-\theta_i\right) \left( \htheta^{(2)}(s,\ell_2)_j-\theta_j\right)  \!\bra{b_\ell}\!\psi_{\theta,\varphi}^{\otimes n_2}\!\ket{b_\ell}\\
    &= \sum_{\ell_1,\ell_2} \left(\htheta^{(1)}(s,\ell_1)_i-\theta_i\right) \left( \htheta^{(2)}(s,\ell_2)_j-\theta_j\right)  \!\bra{b_{\ell_1}}\!\psi_{\theta,\varphi}\!\ket{b_{\ell_1}} \bra{b_{\ell_2}}\!\psi_{\theta,\varphi}\!\ket{b_{\ell_2}} = O(\norm{\Delta\phi}^2)\cdot O(\norm{\Delta\phi}^2) = O(\norm{\Delta\phi}^4). 
    \end{split}
\end{equation}
Finally, combining \eqref{eq:c7} and \eqref{eq:c8}, 
\begin{align}
    \trace( W  V_{U,s}) 
    &= \frac{1}{n_2} \left(\sum_{ij} W_{ji}\sum_{\ell_1} \left(\htheta^{(1)}(s,\ell_1)_i-\theta_i\right) \left( \htheta^{(1)}(s,\ell_1)_j-\theta_j\right)  \!\bra{b_{\ell_1}}\!\psi_{\ctheta,\cU}\!\ket{b_{\ell_1}} + O(\norm{\Delta\phi})\right) + O(\norm{\Delta\phi}^4) \nonumber\\
    &= \frac{1}{n_2}\big(C_{\rm H}(\psi_{\ctheta,\cvarphi},W^*) + O(\|\Delta\phi\|)\big) + O(\|\Delta\phi\|^4). 
\end{align}

\section{Continuity of QCRB and HCRB}

Here we show the continuity of QCRB and HCRB, that is  
\begin{equation}
    \lim_{\ctheta\rightarrow \theta} C_{\rm H,F}(\rho_{\ctheta},W) = C_{\rm H,F}(\rho_{\theta},W). 
\end{equation}

First, the continuity of the QCRB is straightforward because $J(\rho_\theta)$ can be explicitly expressed as 
\begin{equation}
    J(\rho_\theta)_{ij} = 2 \sum_{k,k':\lambda_{k,\theta} + \lambda_{k',\theta} > 0} \frac{\Re[\braket{\psi_{k,\theta}|\partial_i\rho_\theta|\psi_{k',\theta}}\braket{\psi_{k',\theta}|\partial_j\rho_\theta|\psi_{k,\theta}}]}{\lambda_{k,\theta} + \lambda_{k',\theta}}.
\end{equation}
Since we assume the rank $r$ of $\rho_\theta$ does not change, the continuity of the QCRB directly follows from the continuity of the derivatives. 

Next, we show the continuity of the HCRB. 
\begin{equation}
     C_{\rm H}(\rho_\theta,W):=  \min_{\vX:=\{X_i\}_{i=1}^m} \!\trace(W \Re[Z(\vX)]) +  \big\|\sqrt{W}  \Im[Z(\vX)]  \sqrt{W}\big\|_1, 
\end{equation}
where $X_i$ are Hermitian and satisfy
\begin{equation}
    \trace\bigg(X_i \frac{\partial \rho_\theta}{\partial \theta_j}  \bigg) = \delta_{ij}, \quad Z(\vX)_{ij} = \trace(\rho_\theta X_i X_j). 
\end{equation}
Assume $C_{\rm H}(\rho_\theta,W)$ is bounded in the neighborhood of a local point $\theta$ we consider. 
Since $\rho_\theta$ is rank-$r$ and let $\Pi^{(\theta)}$ be a projection onto its support and $\Pi^{(\theta)}_\perp = \id - \Pi^{(\theta)}$. The values of $X_i$ in $\Pi^{(\theta)}_\perp X_i \Pi^{(\theta)}_\perp$ do not contribute to the target function. So we restrict $X_i$ to be in the $2dr-r^2$-dimensional subspace where $\Pi^{(\theta)}_\perp X_i \Pi^{(\theta)}_\perp = 0$. Let $\{G_i^{(\theta)}\}_{i=1}^{2dr-r^2}$ be a continuous, orthonormal basis of this subspace such that $X_i$ can be expressed as $X_i = \sum_i (x_i)_j G_j^{(\theta)}$ where $x_i \in \bR^{2dr-r^2}$. Let $X = (x_1\;x_2\;\cdots\;x_m) \in \bR^{(2dr-r^2)\times m}$. Let 
\begin{equation}
    (S_\theta)_{ij} = \trace(G_i^{(\theta)} G_j^{(\theta)} \rho_\theta),\quad (D_\theta)_{ij} = \trace(G_i^{(\theta)} \partial_j \rho_\theta). 
\end{equation}
Then the constraint on $X_i$ can be written as 
\begin{equation}
    Z[\vX] = X^\top S_\theta X,\quad X^\top D_\theta = \id_m. 
\end{equation}
The HCRB can be rewritten as~\cite{albarelli2019evaluating}
\begin{gather}
    C_{\rm H}(\rho_\theta,W) = \min_{X,V} \trace(W V),\text{~~~~subject~to}\\
    \begin{pmatrix}
        V & X^\top \sqrt{S_\theta} \\
        \sqrt{S_\theta} X & \id 
    \end{pmatrix} \succeq 0,\quad X^\top D_\theta = \id_m,\quad V \in \bR^{m\times m} \text{ is real symmetric.} 
\end{gather}
Now $S_\theta$, $\sqrt{S_\theta}$, $D_\theta$ all vary continuously with $\theta$. 
Note that $\Re[S_\theta]$ is bounded below by a positive constant
in some neighborhood of the point under consideration. 
Regularity condition~(3) gives
$\rho_\theta \succeq c\Pi^{(\theta)}$ for a locally uniform
constant $c>0$, and for any real vector $v$ the Hermitian operator
$G=\sum_i v_iG_i^{(\theta)}$ satisfies
$\Pi_\perp^{(\theta)}G\Pi_\perp^{(\theta)}=0$, so
\begin{align}
    v^\top \Re[S_\theta]v
    = \trace(\rho_\theta G^2) \geq c\,\trace(\Pi^{(\theta)}G^2) \geq \frac{c}{2}\trace(G^2)
     = \frac{c}{2}\|v\|^2.
\end{align}
Below we restrict the domain of $\theta$ to this neighborhood. Also note that $D_\theta$ is full-rank because the linear independence of $\partial_j \rho_\theta$. 

We first prove lower semicontinuity. 
Consider an arbitrary sequence as a subset of $\Theta$ converging to $\theta$,
then let $\theta^{[k]} \rightarrow \theta$ be a subsequence along which
$C_{\rm H}(\rho_{\theta^{[k]}},W)$ converges to the liminf
of the original sequence, and let $(V^{[k]},X^{[k]})$ be the optimal, feasible solutions
for $C_{\rm H}(\rho_{\theta^{[k]}},W)$.
$V^{[k]}$ is bounded because $\trace(V^{[k]}) \leq \norm{W^{-1}} \trace(W V^{[k]}) = \norm{W^{-1}} C_{\rm H}(\rho_{\theta^{[k]}},W) $.
$X^{[k]}$ is bounded because
\begin{align}
    \frac{c}{2}\|X^{[k]}\|_{\mathrm 2}^2
    \leq \trace\bigl(
        (X^{[k]})^\top
        \Re[S_{\theta^{[k]}}]
        X^{[k]}
    \bigr) = \trace\bigl(
        (X^{[k]})^\top
        S_{\theta^{[k]}}
        X^{[k]}
    \bigr) 
    \leq \trace(V^{[k]}).
\end{align} 
Then according to the Bolzano-Weierstrass theorem, there exists a subsequence of the entire sequence such that $(V^{[k]},X^{[k]})$ converges to $(V^{\diamond},X^{\diamond})$. $(V^{\diamond},X^{\diamond})$ is feasible (satisfying the constraint) at the point $\theta$. Thus 
\begin{equation}
    C_{\rm H}(\rho_\theta,W) \leq \trace(W V^{\diamond}) ~~\Rightarrow~~  C_{\rm H}(\rho_\theta,W) \leq \liminf_{\theta' \rightarrow\theta} C_{\rm H}(\rho_{\theta'},W). 
\end{equation}

Next, we prove upper semicontinuity. 
Let $(V,X)$ be an optimal feasible solution at $\theta$ and let $\widehat{V} = V + \eta \id$ for some $\eta > 0$. Then $(\widehat{V},X)$ is feasible at $\theta$. 
Let
\begin{equation}
    X' = X + D_{\theta'} (D_{\theta'}^\top D_{\theta'})^{-1} (\id - X^\top D_{\theta'})^\top,
\end{equation}
which satisfies $D_{\theta'}^\top X' = D_{\theta}^\top X = \id$. For sufficiently close $\theta'$, $(\widehat{V},X')$ is also feasible at $\theta'$ because of the continuity of
$S_{\theta'}$, $D_{\theta'}$ and the fact that $D_{\theta'}$ is full-rank (for all $\theta' \in \Theta$, due to regularity condition (6)). The bound
$\widehat{V}-X^\top S_\theta X\succeq\eta\id$ then implies
$\widehat{V}-(X')^\top S_{\theta'}X'\succeq(\eta/2)\id$
for sufficiently close $\theta'$. 
Then 
\begin{equation}
    C_{\rm H}(\rho_{\theta'},W)  \leq \trace(W \widehat{V}) = \trace(W V) + \eta \trace(W) \leq C_{\rm H}(\rho_{\theta},W) + \eta \trace(W), 
\end{equation}
and taking $\eta \rightarrow 0$, we have 
\begin{equation}
\limsup_{\theta'\rightarrow \theta} C_{\rm H}(\rho_{\theta'},W) \leq C_{\rm H}(\rho_{\theta},W). 
\end{equation}

\section{Example: Semidefinite programming for optimal probes in channel estimation}
\label{app:channel-sdp}

We illustrate an example application of \thmref{thm:purification} to the optimization of probe
states in multi-parameter channel estimation. 
The goal is to demonstrate how an semidefinite programming (SDP) algorithm for finding optimal (ancilla-assisted) probe states for unitary estimation~\cite{matsumoto2002new,gorecki2020optimal} can be converted into a corresponding SDP finding optimal (ancilla-assisted) probe states for noisy channel estimation using \thmref{thm:purification}. This example shows the potential power of \thmref{thm:purification} in converting existing results about noiseless channel estimation to noisy ones. 

Note that SDP algorithms for optimizing the HCRB over ancilla-assisted probes for general channels was already found in prior work by Hayashi and Ouyang~\cite{hayashi2024finding}. Our purpose here is to provide an alternative purification-based derivation that illustrates the power of \thmref{thm:purification}. The hope is our calculations here can inspire the development of new applications of \thmref{thm:purification} in future works, e.g. when multiple channels applications, or quantum error correction protocols, are considered. 

Consider a smooth family of channels $\mathcal E_\theta$ mapping a
$d$-dimensional system $S$ to itself, with $m$ identifiable parameters of interest. $\{\partial_j \mE_\theta\}_j$ are linearly independent. The Kraus rank of $\mE_\theta$ is fixed to be a known constant $r$ in the domain of $\theta$ and the corresponding Kraus operators are also smooth. All quantities
below are evaluated at a fixed operating point, taken without loss of generality to be $\theta=0$.
For a cost matrix $W\succ0$, define
\begin{equation}
 C_{\rm H}(\mathcal E_\theta,W)
 :=\inf_{\ket{\psi}_{SA}}
 C_{\mathrm H}\!\left(
 (\mathcal E_\theta\otimes\mathrm{id}_A)
 (\ket{\psi}\bra{\psi}_{SA}),W\right).
 \label{eq:app-channel-cost}
\end{equation}
The ancillary system $A$ is noiseless and sufficiently large. 
Pure inputs suffice, since a mixed input can be purified into an
additional accessible ancilla. 
Note that we assume the existence of a smooth Kraus representation, the fixed Kraus rank and the identifiability of the channel
parameters. Our goal below is to show how \eqref{eq:app-channel-cost} can be solved using an SDP. 

We will review previous results in the first two subsections, explain a special case in the third subsection and end with the general SDP result.  

\subsection{Formulation of the Matsumoto bound for pure state estimation}
\label{app:channel-matsumoto}

For a pure-state model $\ket{\psi_\theta}\in\mathcal H$, the Matsumoto
form of the HCRB~\cite{matsumoto2002new,gorecki2020optimal} is
\begin{equation}
 C_{\mathrm H}(\psi_\theta,W)
 =\min_{\{\ket{x_i}\}_{i=1}^m}
 \operatorname{Tr}(WZ),
 \qquad Z_{ij}=\langle x_i|x_j\rangle,
 \label{eq:app-matsumoto}
\end{equation}
where $\ket{x_i}\in\mathcal H\oplus\mathbb C^m$ satisfy
\begin{equation}
 \langle x_i|\psi_\theta\rangle=0,
 \qquad 2\Re [\langle x_i|\partial_j\psi_\theta\rangle]
 =\delta_{ij},
 \qquad \Im [Z]=0.
 \label{eq:app-matsumoto-constraints}
\end{equation}
The extra space permits the real-Gram-matrix formulation and may be
embedded in a sufficiently large ancillary space. 
Equivalently, 
\begin{equation}
 C_{\mathrm H}(\psi_\theta,W)
 =\min_{\{\ket{c_i}\}_{i=1}^m}
 \operatorname{Tr}\!\left(W(D^\top D)^{-1}\right),
 \label{eq:app-orthonormal-form}
\end{equation}
where $\ket{c_i}\in\mathcal H\oplus\mathbb C^m$ and $D$ is invertible, satisfying
\begin{equation}
 D_{ij}=2\Re [ \langle c_i|\partial_j\psi_\theta\rangle ],
 \qquad \langle c_i|c_j\rangle=\delta_{ij},
 \qquad \langle c_i|\psi_\theta\rangle=0. 
 \label{eq:app-D}
\end{equation}
To see the equivalence, choose an orthonormal basis of the real span of
$\{\ket{x_i}\}$ and write
$\mX_{ai}=\langle c_a|x_i\rangle\in\mathbb R$.
Then $\mX^\top D=\id_m$, and 
\begin{equation*}
 Z=\mX^\top \mX=(D^\top D)^{-1}.
\end{equation*}
Conversely, given a feasible solution of \eqref{eq:app-D}, the vectors
$\ket{x_i}=\sum_a(D^{-1})_{ia}\ket{c_a}$ satisfy
\eqref{eq:app-matsumoto-constraints}.

\subsection{SDP for unitary channel estimation}
\label{app:channel-unitary}

Here we review a result in~\cite{gorecki2020optimal}. 
First consider $\mathcal U_\theta(\rho)=U_\theta\rho U_\theta^\dagger$,
where
\begin{equation}
 U_\theta=\exp\!\bigg(-i\sum_{j=1}^m\theta_jG_j\bigg).
\end{equation}
We take the Hermitian generators to be traceless and orthonormal,
$\operatorname{Tr}(G_iG_j)=\delta_{ij}$. Note that any linearly independent generators $\{G_i\}_{i=1}^m$
can be put in this form by a linear reparametrization, accompanied by
the corresponding transformation of $W$.
We also complete $\{G_j\}_{j=1}^m$ to an orthonormal basis of traceless
Hermitian operators with $\{R_\nu\}_{\nu=m+1}^{d^2-1}$ for later use.

Then we can use the following SDP to find the optimal HCRB: 
\begin{equation}
 C_{\rm H}(\mathcal U_\theta,W)
 =\frac m4\min_{\{G_j^L\},\{R_\nu^L\},K,w}w,
 \label{eq:app-unitary-sdp}
\end{equation}
subject to
\begin{align}
 &\id_{m+1}\otimes\frac{\id_d}{d}
 +\sum_{j=1}^m(G_j^L)^\top \otimes G_j
 +\sum_{\nu=m+1}^{d^2-1}(R_\nu^L)^\top \otimes R_\nu
 \succeq0,
 \label{eq:app-unitary-Q}\\
 &\Gamma_{ij}=\Im[G_j^L]_{i0},
 \qquad i,j=1,\ldots,m,
 \label{eq:app-unitary-Gamma}\\
 &\Gamma W^{-1/2}=(\Gamma W^{-1/2})^\top \succeq 0,
 \notag% \label{eq:app-unitary-polar}
 \\
 & \begin{pmatrix}
 mw&\operatorname{vec}(K)^\top \\
 \operatorname{vec}(K)&\id_{m^2}
 \end{pmatrix}\succeq0,
 \qquad
 \begin{pmatrix}K&\id_m\\\id_m&\Gamma W^{-1/2}\end{pmatrix}\succeq0.
 \notag% \label{eq:app-unitary-epigraph}
\end{align}
Here $G_j^L,R_\nu^L$ are Hermitian $(m+1)\times(m+1)$ matrices,
$K$ is real symmetric, and $w\in\mathbb R$.
The operation $\operatorname{vec}$ stacks the columns of $K$ into a
column vector such that $\operatorname{vec}(K)^\top  \operatorname{vec}(K) = \trace(K^\top K) = \trace(K^2)$. All the constraints are affine equalities or linear
matrix inequalities.

To prove the result, we first show the optimal HCRB can be written as the equivalent (but not apparently
semidefinite) optimization
\begin{equation}
 C_{\rm H}(\mathcal U_\theta,W)
 =\frac14\min_{\{G_j^L\},\{R_\nu^L\}}
 \operatorname{Tr}\!\left(W(\Gamma^\top \Gamma)^{-1}\right),
 \label{eq:app-unitary-nonlinear}
\end{equation}
subject to \eqref{eq:app-unitary-Q} and
\eqref{eq:app-unitary-Gamma}.
Let $\ket{c_0}=\ket{\psi}_{SA}$ be the input probe state. 
$\ket{\partial_j\psi_\theta}\!\big|_{\theta = 0} = -i(G_j\otimes \id_A)\ket{c_0}$.
From \eqref{eq:app-orthonormal-form}, assume $A$ is sufficiently large, we have 
\begin{equation}
\label{eq:app-channel-D}
C_{\rm H}(\mathcal U_\theta,W)  = \min_{\{\ket{c_i}\}_{i=0}^m,\ket{c_i}\in \mH_{SA}} \trace(W (D^\top D)^{-1}), 
\end{equation}
where 
\begin{equation*}
   D_{ij} = 2\Im[\braket{c_i|G_j\otimes \id_A|c_0}], ~~\forall i,j\in[m], \quad \braket{c_i|c_j} = \delta_{ij}, ~~\forall i,j\in[m]\cup\{0\}. 
\end{equation*}

Introduce an artificial Hilbert space which we call the logical space $\mathcal H_L\cong\mathbb C^{m+1}$ and let 
\begin{equation*}
 Q:=\operatorname{Tr}_A\!\left(
 \sum_{k,l=0}^m\ket{k}\bra{l}_L\otimes\ket{c_k}\bra{c_l}
 \right).
 \label{eq:app-Q-frame}
\end{equation*}
Let $(G_i^L)_{kl} := \bra{c_k}G_i\otimes \id_A\ket{c_l}$ be the effective Hamiltonian in the logical space. $D_{ij} = 2\Im[(G_j^L)_{i0}]$. We have 
\[
(G_i^L)^\top  = \trace_S (Q (\id_{m+1} \otimes G_i)),\quad i = 1,\ldots,m. 
\]
Similarly, 
$(R_i^L)^\top  = \trace_S (Q (\id_{m+1} \otimes R_i))$. 
$Q$ can be written as 
\[
Q = \id_{m+1} \otimes \frac{\id_d}{d} + \sum_{i=1}^m (G_i^L)^\top  \otimes G_i + \sum_{i=m+1}^{d^2-1} (R_i^L)^\top  \otimes R_i . 
\]
Any set of Hermitian operators $\{G_i^L\}_{i=1}^m$ satisfying $Q \succeq 0$ for some choice of $\{R_i^L\}_{i=m+1}^{d^2-1}$ induces a corresponding choice of $\{\ket{c_i}\}_{i=0}^m$ such that $(G_i^L)_{kl} := \bra{c_k}G_i\otimes \id_A\ket{c_l}$. To see this, we can first obtain a purification of $Q$, $\ket{Q} = \sum_{i=0}^m \ket{k}_L \otimes \ket{c_k}_{SA}$ and $\ket{c_k}$ satisfies the required condition. Similarly, for any $\{\ket{c_i}\}_{i=0}^m$, $Q$ defined as above is always non-negative. Therefore, we can rewrite \eqref{eq:app-channel-D} as \eqref{eq:app-unitary-nonlinear}, taking $\Gamma = D/2$. 

To see \eqref{eq:app-unitary-nonlinear} is the same as \eqref{eq:app-unitary-sdp}, we first note that we can add a condition that $\Gamma \sqrt{W^{-1}} \succeq 0$ to \eqref{eq:app-unitary-nonlinear} without loss of generality. First, we can always find an orthogonal matrix $O$ such that $O \Gamma \sqrt{W^{-1}} \succeq 0$. Next, we note that replacing $\Gamma$ with $O \Gamma$ is equivalent to a basis rotation of $\{\ket{k}_L\}_{i=1}^m $ which does not affect the validity of $Q \succeq 0$. Finally, $\Gamma^\top  \Gamma$ is invariant under $O$. After adding $\Gamma \sqrt{W^{-1}} \succeq 0$, we note that 
\[
\min\trace(W(\Gamma^\top  \Gamma)^{-1}) = \min \trace(K^2) = m\min w, \text{ where}
\]
\[
\begin{pmatrix}
K & \id_m \\
\id_m & \Gamma\sqrt{W^{-1}}
\end{pmatrix} \succeq 0 \Leftrightarrow K \geq (\Gamma\sqrt{W^{-1}})^{-1}, \quad 
\begin{pmatrix}
 mw&\operatorname{vec}(K)^\top \\
 \operatorname{vec}(K)&\id_{m^2}
 \end{pmatrix}\succeq0 \Leftrightarrow mw\geq\operatorname{Tr}(K^2). 
\]
The solution $\ket{c_0}$ found from ($Q$ from) \eqref{eq:app-unitary-sdp} is the optimal input state we want.

\begin{figure}[tb]
\centering

\begin{tikzpicture}[
    >=Latex,
    font=\small,
    wire/.style={line width=0.8pt},
    gate/.style={
        draw,
        rounded corners=1.5pt,
        minimum width=1.25cm,
        minimum height=0.85cm,
        align=center,
        fill=white
    }
]

% ============================================================
% LEFT PANEL: shift everything LEFT
% ============================================================
\begin{scope}[xshift=-1.0cm]

\node[font=\bfseries] at (2.8,2.5)
    {(a) Original channel estimation};

\node (psiL) at (0.55,1.35) {$\ket{\psi}_{SA}$};

\draw[wire] (1.20,1.70) -- (4.55,1.70);
\node[left]  at (1.20,1.70) {$A$};
\node[right] at (4.55,1.70) {$A$};

\draw[wire] (1.20,1.00) -- (2.30,1.00);
\node[left] at (1.20,1.00) {$S$};

\node[gate] (channel) at (2.95,1.00)
    {$\mathcal E_\theta$};

\draw[wire] (channel.east) -- (4.55,1.00);
\node[right] at (4.55,1.00) {$S$};

\node at (2.9,0.25)
    {$\rho_\theta=
    (\mathcal E_\theta\otimes\mathrm{id}_A)(\psi_{SA})$};

\end{scope}

% ============================================================
% THEOREM 1: fixed in the middle
% ============================================================

\node[font=\normalsize] at (5.75,1.75)
    {\thmref{thm:purification}};

\draw[<->,line width=0.9pt]
    (5.1,1.35) -- (6.4,1.35);

% ============================================================
% RIGHT PANEL: shift everything RIGHT
% ============================================================
\begin{scope}[xshift=2.0cm]

\node[font=\bfseries] at (8.5,2.5)
    {(b) Purified formulation};

\node (psiR) at (6.05,1.35)
    {$\ket{\psi}_{SA}$~~~};

\draw[wire] (6.70,1.70) -- (11.15,1.70);
\node[left]  at (6.70,1.70) {$A$};
\node[right] at (11.15,1.70) {$A$};

\draw[wire] (6.70,1.00) -- (7.45,1.00);
\node[left] at (6.70,1.00) {$S$};

\node (env) at (6.05,0.25) {$\ket{0}_E$~~~};
% \draw[wire] (6.70,0.25) -- (Vin.west |- 0,0.25);
\node[left] at (6.70,0.25) {$E$};

\node[
    gate,
    minimum height=1.50cm,
    minimum width=1.25cm
] (Vin) at (8.10,0.625)
    {$U_{\rm in}$};

\draw[wire]
    (7.45,1.00) -- (Vin.west |- 0,1.00);

\draw[wire]
    (7.45,0.25) -- (Vin.west |- 0,0.25);

\node[
    gate,
    minimum height=1.50cm,
    minimum width=1.55cm
] (U) at (9.75,0.625)
    {$U_{\theta,\varphi}$};

\draw[wire]
    (Vin.east |- 0,1.00) -- (U.west |- 0,1.00);

\draw[wire]
    (Vin.east |- 0,0.25) -- (U.west |- 0,0.25);

\draw[wire]
    (U.east |- 0,1.00) -- (11.15,1.00);

\draw[wire]
    (U.east |- 0,0.25) -- (11.15,0.25);

\draw[wire] (6.70,0.25) -- (7.45,0.25);

\node[right] at (11.15,1.00) {$S$};
\node[right] at (11.15,0.25) {$E$};

\draw[dashed,->]
    (10.72,0.25)
    to[out=25,in=-75]
    (10.72,0.90);

\node[right] at (10.78,0.62)
    {$\operatorname{Tr}_E$};

\node at (8.80,-0.45)
    {$\ket{\Psi_{\theta,\varphi}}_{SAE}\ \text{pure},
    \qquad
    \varphi:\text{ nuisance parameters}$};

\end{scope}

\end{tikzpicture}

\caption{
Purification-based formulation of channel estimation.
(a) An ancilla-assisted probe is sent through the noisy channel
$\mathcal E_\theta$.
(b) A Stinespring dilation represents the same output through a pure-state
model on $SAE$, with the environmental freedom $\varphi$ treated as
nuisance parameters.
}
\label{fig:channel-purification}

\end{figure}

\subsection{SDP for a parameterized unitary after fixed noise}
\label{app:channel-fixed-noise}

Consider next the estimation of quantum channels 
\begin{equation}
 \mathcal E_\theta=\mathcal U_\theta\circ\mathcal E_0.
\end{equation}
Let $r$ be the Kraus rank of $\mathcal E_0$, and choose an
$r$-dimensional environment $E$ and a Stinespring isometry
\begin{equation}
 V_{\mathrm{in}}=U_{\mathrm{in}}(\id_S\otimes\ket{0}_E),
 \qquad
 \mathcal E_0(\rho)=\operatorname{Tr}_E
 (V_{\mathrm{in}}\rho V_{\mathrm{in}}^\dagger).
\end{equation}
Let $m^*=m+r^2-1$, and choose a traceless Hermitian basis
$\{H_a\}_{a=1}^{r^2-1}$ on $E$, normalized by
$r\operatorname{Tr}(H_aH_b)=\delta_{ab}$.
A purification with environmental nuisance parameters is
\begin{gather}
 \ket{\Psi_{\theta,\varphi}}
 =(U_{\theta,\varphi}V_{\mathrm{in}}\otimes \id_A)
 \ket{\psi}_{SA},\notag \\
 U_{\theta,\varphi}
 =\exp\!\left(-i\sum_{j=1}^m\theta_j(G_j\otimes \id_E)
               -i\sum_{a=1}^{r^2-1}\varphi_a(\id_S\otimes H_a)\right).
 \label{eq:app-fixed-noise-dilation}
\end{gather}
The environment rotations leave the channel output unchanged.
For a regular probe, \thmref{thm:purification} identifies its mixed-state HCRB with
the HCRB of this pure-state model with weight
\begin{equation*}
 W^*=\begin{pmatrix}W&0\\0&0\end{pmatrix}.
\end{equation*}
We can choose suitable $H_a$ and perform suitable reparameterization of $\theta$ (which is equivalent to a congruence transformation on the cost matrix $W$ for our purpose) such that 
\begin{equation}
\label{eq:app-orthonormal}
r\trace(G_iG_j) = \delta_{ij},\quad d\trace(H_a H_b) = \delta_{ab}. 
\end{equation}
Let $\{R_\nu\}_{\nu = 1}^{d^2r^2-m^*-1}$ be the rest of the traceless Hermitian operators which, together with $\{G_i\otimes \id_E\}_i$ and $\{\id_S\otimes H_a\}_a$, form a complete orthonormal basis of traceless Hermitian operators acting on $\mH_{SE}$. 
According to \thmref{thm:purification}, 
\begin{equation}
\label{eq:app-purification}
\inf_{\psi_{SA}: \text{full Schmidt rank}} C_{\rm H}( (\mE_{\theta} \otimes \id_A)(\psi_{SA}), W)= \inf_{\psi_{SA}: \text{full Schmidt rank}} C_{\rm H}( (\mU_{\theta,\varphi}\circ \mU_{\mathrm{in}} \otimes \id_A)(\psi_{SA} \otimes \ket{0}_E\bra{0}_E), W^*). 
\end{equation}
Here $r$ being the Kraus rank ensures for any $\psi_{SA}$ with a full Schmidt rank, i.e. $\trace_A(\psi_{SA})$ is full-rank, $(\mE_\theta\otimes \id)(\psi_{SA})$ is a rank-$r$ mixed state. 
Also, note that although \appref{app:regularity} requires regularity conditions~(3), (4) and (5) to be satisfied for the application of \thmref{thm:purification}---they were not all necessary here. First, only local identifiability (i.e. regularity condition~(6)) is needed for HCRB related calculations at point $\theta = 0$, and identifiability (i.e. regularity condition~(5)), only included earlier for narrative simplicity, can be dropped here. Second, regularity condition (4) and the nondegenerate positive spectrum assumption from regularity condition~(3) were included in \appref{app:regularity} to guarantee the existence of a unique and three times continuously differentiable diagonalization form of the state, while here they can also be dropped because we already have the smooth function $\Psi_{SAE}$ defined after the Stinespring dilation. Finally, regularity condition~(6) is satisfied because every full-Schmidt-rank probe preserves the linear independence of the channel derivatives (following from the channel's identifiability).

Next, we extend the domain of $\psi_{SA}$ to all possible states, that leads to
\begin{equation}
\label{eq:app-purification-1}
C_{\rm H}( \mE_{\theta}, W) = \inf_{\psi_{SA}} C_{\rm H}( (\mE_{\theta} \otimes \id_A)(\psi_{SA}), W)= \inf_{\psi_{SA}: \text{full Schmidt rank}} C_{\rm H}( (\mU_{\theta,\varphi}\circ \mU_{\mathrm{in}} \otimes \id_A)(\psi_{SA} \otimes \ket{0}_E\bra{0}_E), W^*). 
\end{equation}
This is possible because for any input state $\psi_{SA}$ that is rank-deficient we can replace $\ket{\psi_{SA}}$ with a new full-Schmidt-rank state written as 
\begin{equation}
\label{eq:full-schmidt}
 \ket{\psi_{SA}}  \rightarrow \sqrt{1-\epsilon} \ket{\psi_{SA}} + \sqrt{\epsilon} \ket{\psi^\perp_{SA}},
\end{equation}
where $\ket{\psi^\perp_{SA}}$ and $\ket{\psi_{SA}}$ have orthogonal supports on $A$. Specifically, when $X_i$ is an optimal solution of the HCRB for the original state $(\mE_\theta\otimes \id_A)(\psi_{SA})$, we can find a feasible choice (satisfying \eqref{eq:cond}) by 
\[
    X_i \rightarrow \frac{X_i}{1-\epsilon},\qquad Z(\vX)_{ij} = \trace((\mE_\theta\otimes \id_A)(\psi_{SA})X_iX_j) \rightarrow \frac{Z(\vX)_{ij}}{1 - \epsilon}. 
\]
The HCRB using the new input state is upper bounded by $1/(1-\epsilon)$ times the HCRB using the original input state. 
Thus, $\inf_{\psi_{SA}: \text{full Schmidt rank}} C_{\rm H}( (\mE_{\theta} \otimes \id_A)(\psi_{SA}), W) \leq \inf_{\psi_{SA}} C_{\rm H}( (\mE_{\theta} \otimes \id_A)(\psi_{SA}), W)$. The converse direction is trivial. 

In order to provide an SDP solving $C_{\rm H}( \mE_{\theta}, W)$ (\eqref{eq:app-purification-1}), we use the right-hand side of \eqref{eq:app-purification-1} and apply a similar derivation as in the previous discussion. Two modifications are required, as detailed below. 

\begin{enumerate}[(1), wide, labelwidth=!, labelindent=0pt]
    \item First, we revisit and modify the vector form of the Matsumoto bound (as derived in \cite{matsumoto2002new,gorecki2020optimal} using Naimark's extension theorem), Eqs.~(\ref{eq:app-matsumoto})--(\ref{eq:app-D}), where $W$ is replaced with a singular matrix $W^*$. Assume $\Psi_{\theta,\varphi}$ is locally identifiable, i.e. $J(\Psi_{\theta,\varphi}) \succ 0$. We have 
\begin{equation}
 C_{\mathrm H}(\Psi_{\theta,\varphi},W^*)
 =\min_{\{\ket{x_i}\}_{i=1}^m \subseteq \mathcal H\oplus\mathbb C^m }
 \operatorname{Tr}(WZ),
 \qquad Z_{ij}=\langle x_i|x_j\rangle,
 \label{eq:app-matsumoto--}
\end{equation}
where 
\begin{equation*}
 \langle x_i|\Psi_\theta\rangle=0,
 \quad 2\Re [\langle x_i|\partial_{\theta_j}\Psi_{\theta,\varphi}\rangle]
 =\delta_{ij},
 \quad 2\Re [\langle x_i|\partial_{\varphi_a}\Psi_{\theta,\varphi}\rangle]
 =0,
 \quad \Im [Z]=0,\quad \forall i,j\in[m],a\in[r^2-1]. 
\end{equation*}
Equivalently, 
\begin{equation}
 C_{\mathrm H}(\Psi_{\theta,\varphi},W^*)
 =\min_{\{\ket{c_i} \}_{i=1}^m \subseteq \mathcal H\oplus\mathbb C^m }
 \operatorname{Tr}\!\left(W(D^\top D)^{-1}\right),
 \label{eq:app-orthonormal-form--}
\end{equation}
where $D$ is invertible, satisfying
\begin{equation*}
 D_{ij}=2\Re [ \langle c_i|\partial_{\theta_j}\Psi_{\theta,\varphi}\rangle ],
 \quad 2\Re [ \langle c_i|\partial_{\varphi_a}\Psi_{\theta,\varphi}\rangle ] = 0, 
 \quad \langle c_i|c_j\rangle=\delta_{ij},
 \quad \langle c_i|\Psi_{\theta,\varphi}\rangle=0,\quad \forall i,j\in[m],a\in[r^2-1]. 
\end{equation*}
Compared to the case without nuisance parameters, the condition $2\Re [ \langle c_i|\partial_{\varphi_a}\Psi_{\theta,\varphi}\rangle ] = 0$ is added. Let 
\[
\ket{c_0} = ((U_{\mathrm{in}})_{SE}\otimes \id_A) (\ket{\psi_{SA}} \otimes \ket{0}_E) = ((V_{\mathrm{in}})_{S\rightarrow SE}\otimes \id_A) \ket{\psi_{SA}}
\]
denote the input state after $U_{\mathrm{in}}$, we have 
\[
\Im[\braket{c_i|\id_{SA} \otimes (H_a)_{E}|c_0}] = 0,\quad \forall i \in [m], a \in [r^2-1]. 
\]
\item If $\ket{c_0}$ can be chosen arbitrarily, then following the same argument as in the previous subsection, we have 
\begin{equation}
 \inf_{\Psi_{SAE}} C_{\rm H}( (\mU_{\theta,\varphi}\circ \mU_{\mathrm{in}} \otimes \id_A)(\Psi_{SAE}), W^*)
 =\frac m4\min_{\{G_j^L\},,\{H_a^L\},\{R_\nu^L\},K,w}w,
\end{equation}
subject to
\begin{align*}
 &Q = \id_{m+1}\otimes\frac{\id_{dr}}{dr}
 +\sum_{j=1}^m(G_j^L)^\top \otimes (G_j\otimes \id_E)
 +\sum_{a=1}^{r^2-1}(H_a^L)^\top \otimes (\id_S \otimes H_a)
 +\sum_{\nu=1}^{d^2r^2-m^*-1}(R_\nu^L)^\top \otimes R_\nu
 \succeq0,
 \\
 &\Gamma_{ij}=\Im[G_j^L]_{i0},
 \qquad i,j=1,\ldots,m,
 \\
 & \Im[H_a^L]_{i0} = 0,\qquad i=1,\ldots,m,\,a=1,\ldots,r^2-1,\\
 &\Gamma W^{-1/2}=(\Gamma W^{-1/2})^\top \succeq 0,
 \notag
 \\
 & \begin{pmatrix}
 mw&\operatorname{vec}(K)^\top \\
 \operatorname{vec}(K)&\id_{m^2}
 \end{pmatrix}\succeq0,
 \qquad
 \begin{pmatrix}K&\id_m\\\id_m&\Gamma W^{-1/2}\end{pmatrix}\succeq0.
 \notag
\end{align*}
Here $G_j^L,H_a^L,R_\nu^L$ are Hermitian $(m+1)\times(m+1)$ matrices,
$K$ is real symmetric, and $w\in\mathbb R$.

To impose the input state constraint, we note that 
\[
V_{\mathrm{in}}^\dagger  \bra{0}_L Q \ket{0}_L V_{\mathrm{in}} = V_{\mathrm{in}}^\dagger \trace_A(\ket{c_0}\bra{c_0}) V_{\mathrm{in}} = \trace_A(\psi_{SA})
\]
can be an arbitrary density matrix. Since the positivity of $Q$ guarantees the positivity of $V_{\mathrm{in}}^\dagger  \bra{0}_L Q \ket{0}_L V_{\mathrm{in}}$, we need to add only one additional constraint: 
\[
\trace( V_{\mathrm{in}}^\dagger  \bra{0}_L Q \ket{0}_L V_{\mathrm{in}} ) = 1, 
\]
to fix the input state to be $\psi_{SA} \otimes \ket{0}_E \bra{0}_E$.
\end{enumerate}

To summarize, from the above discussion, we have when $\mE_\theta = \mU_\theta \circ \mE_0$,
\begin{equation}
\label{eq:app-final}
C_{\rm H}( \mE_{\theta}, W) =\frac m4\min_{\{G_j^L\},\{H_a^L\},\{R_\nu^L\},K,w}w,
\end{equation}
subject to 
\begin{align}
 &\id_{m+1}\otimes\frac{\id_{dr}}{dr}
 +\sum_{j=1}^m(G_j^L)^\top \otimes (G_j\otimes \id_E)
 +\sum_{a=1}^{r^2-1}(H_a^L)^\top \otimes (\id_S \otimes H_a)
 +\sum_{\nu=1}^{d^2r^2-m^*-1}(R_\nu^L)^\top \otimes R_\nu
 \succeq0,
 \label{eq:app-unitary-Q--}
 \\
 &\Gamma_{ij}=\Im[G_j^L]_{i0},
 \qquad i,j=1,\ldots,m,
 \notag
 \\
 & \Im[H_a^L]_{i0} = 0,\qquad i=1,\ldots,m,\,a=1,\ldots,r^2-1,\label{eq:app-nuisance}\\
 &\trace( V_{\mathrm{in}}^\dagger  \bra{0}_L Q \ket{0}_L V_{\mathrm{in}} ) = 1,\label{eq:app-input-0}\\
 &\Gamma W^{-1/2}=(\Gamma W^{-1/2})^\top \succeq 0,
 \notag
 \\
 & \begin{pmatrix}
 mw&\operatorname{vec}(K)^\top \\
 \operatorname{vec}(K)&\id_{m^2}
 \end{pmatrix}\succeq0,
 \qquad
 \begin{pmatrix}K&\id_m\\\id_m&\Gamma W^{-1/2}\end{pmatrix}\succeq0.
 \notag
\end{align}
Here $G_j^L,H_a^L,R_\nu^L$ are Hermitian $(m+1)\times(m+1)$ matrices,
$K$ is real symmetric, and $w\in\mathbb R$. The SDP algorithm follows from the unitary estimation SDP algorithm \eqref{eq:app-unitary-sdp}, with two additional constraints---\eqref{eq:app-nuisance} to treat nuisance parameters and \eqref{eq:app-input-0} to treat the dilation in the input state. 

One \emph{caveat} is that an optimal solution of the SDP in \eqref{eq:app-final} may
correspond to an input probe without full Schmidt rank, in which case \thmref{thm:purification} cannot be applied directly to the recovered probe. Nevertheless, the SDP optimal value still equals the infimum in \eqref{eq:app-purification-1}. In fact, the SDP still
gives the correct infimum. For any feasible SDP solution with $\{|c_i\rangle\}_{i=0}^m$, enlarge the ancilla and define
\begin{align*}
|c_0^{(\epsilon)}\rangle
&=
\sqrt{1-\epsilon}\,|c_0\rangle_{SAE} \oplus 0_{A'}
+
\sqrt{\epsilon}\,|c_{\rm ref}\rangle_{SA'E} \oplus 0_{A},\\
|c_i^{(\epsilon)}\rangle
&=
|c_i\rangle_{SAE} \oplus 0_{A'},\qquad i=1,\ldots,m,
\end{align*}
where $|c_{\rm ref}\rangle_{SA'E} = (V_{\rm in} \otimes \id_{A'}) \ket{\Phi_{SA'}}$ is obtained from a maximally entangled input between $S$ and a new ancillary subspace $A'$ perpendicular to the subspace of $A$. Now the entire ancillary system is $A \oplus A'$. This guarantees $|c_0^{(\epsilon)}\rangle$ is perpendicular to $\ket{c_i^{(\epsilon)}}$ for $i \geq 1$. 
The corresponding probe has full Schmidt rank for every $\epsilon>0$.
Moreover,
\[
D^{(\epsilon)}=\sqrt{1-\epsilon}\,D,
\]
while the nuisance constraints remain satisfied. Hence, by
\eqref{eq:app-orthonormal-form--}, 
\[
C_{\rm H}(\mathcal E_\theta,W)
\le
\frac{1}{1-\epsilon}
\Tr\!\left( W(D^\top D)^{-1}\right)
\le
\frac{1}{1-\epsilon}\frac{m}{4}w .
\]
Taking $\epsilon\to0$ and then minimizing over all feasible SDP
solutions gives
\[
C_{\rm H}(\mathcal E_\theta,W)
\le
\frac{m}{4}\min w .
\]
Thus, \eqref{eq:app-final} exactly evaluates
$C_{\rm H}(\mathcal E_\theta,W)$. Furthermore, there exists a full-Schmidt-rank probe state that achieves the optimal HCRB with an arbitrarily small error. 

\subsection{SDP for general channels}
\label{app:channel-general}

Here we show the SDP algorithm \eqref{eq:app-final} derived above can be easily adapted for solving general channel estimation. 
Consider a general parameterized channel $\mE_\theta(\rho) = \sum_{i=1}^r K_{i,\theta} \rho K_{i,\theta}^\dagger$ where $r$ is the Kraus rank. To calculate the channel HCRB at $\theta = 0$, we show below without loss generality we can assume  
\[
\mE_\theta(\rho) = \trace_E( U_{\theta,\varphi} U_{\rm in} (\rho \otimes \ket{0}_E\bra{0}_E) U_{\rm in}^\dagger U_{\theta,\varphi}^\dagger),
\]
where 
\[
U_{\theta,\varphi} = \exp\left(-i \sum_{a=1}^{r^2-1} \varphi_a \id_S \otimes H_a \right)\exp\left(-i \sum_{i=1}^m \theta_i G_i \right).
\]
$H_a$ are orthonormal traceless Hermitian operators as defined before, and $G_i$ are Hermitian operators. 
In contrast to the previous case, here $G_i$ is not necessarily an operator acting solely on the system $S$ but extends to the environment $E$. Note that only $\mE_0$ and the derivative of $\mE_\theta$ at $\theta = 0$ will enter into the HCRB calculation. Thus we only need to find a suitable choice of $U_{\rm in}$ and $G_i$ such that 
\[
\mE_0(\rho) = \sum_{i=1}^r K_{i,0} \rho K_{i,0}^\dagger = \trace_E( U_{0,0} U_{\rm in} (\rho \otimes \ket{0}_E\bra{0}_E) U_{\rm in}^\dagger U_{0,0}^\dagger) = \trace_E(  U_{\rm in} (\rho \otimes \ket{0}_E\bra{0}_E) U_{\rm in}^\dagger ),
\]
and for $j \in [m]$
\begin{align}
\partial_j \mE_\theta(\rho) \big|_{\theta = 0} 
&= \sum_{i=1}^r \partial_j K_{i,0} \rho K_{i,0}^\dagger + \sum_{i=1}^r K_{i,0} \rho \partial_j K_{i,0}^\dagger \notag\\
&= \trace_E( \partial_j U_{0,0} U_{\rm in} (\rho \otimes \ket{0}_E\bra{0}_E) U_{\rm in}^\dagger U_{0,0}^\dagger) + \trace_E( U_{0,0} U_{\rm in} (\rho \otimes \ket{0}_E\bra{0}_E) U_{\rm in}^\dagger \partial_j U_{0,0}^\dagger) \notag\\
&= \trace_E( (-iG_j) U_{\rm in} (\rho \otimes \ket{0}_E\bra{0}_E) U_{\rm in}^\dagger ) + \trace_E(  U_{\rm in} (\rho \otimes \ket{0}_E\bra{0}_E) U_{\rm in}^\dagger (iG_j) ).\label{eq:app:orgreq}
\end{align}
To satisfy the above conditions, it is sufficient to find $U_{\rm in}$ and $G_i$ such that  $\forall i \in [r]$,
\begin{equation}
\label{eq:app-deriv}
K_{i,0} = \bra{i_E} U_{\rm in} \ket{0_E},\quad 
\partial_j K_{i,\theta} \big|_{\theta = 0} = \bra{i_E} (-i G_j) U_{\rm in} \ket{0_E},~\forall j \in [m]. 
\end{equation}
$U_{\rm in}$ and $V_{\rm in} = U_{\rm in}\ket{0}_E$ can be chosen straightforwardly from the Stinespring dilation. Let
\[
V_{\rm in} = \sum_i K_{i,0} \otimes \ket{i}_E,
\quad 
B_j = \sum_i \partial_j K_{i,\theta} \big|_{\theta = 0} \otimes \ket{i}_E
\]
Taking the derivative of $\sum_i K_i^\dagger K_i = \id_S$, we have 
\[
B_j^\dagger V_{\rm in} + V_{\rm in}^\dagger B_j = 0, 
\]
and $V_{\rm in}^\dagger B_j$ is anti-Hermitian. 
The following choice of $G_j$ satisfies our requirement 
\begin{equation}
-iG_j =  - V_{\rm in}  (B_j^\dagger V_{\rm in})V_{\rm in}^\dagger + (\id_{SE} - V_{\rm in}V_{\rm in}^\dagger) B_j V_{\rm in}^\dagger - V_{\rm in} B_j^\dagger  (\id_{SE} - V_{\rm in}V_{\rm in}^\dagger) 
\end{equation}
We can verify that 
\[
-iG_j V_{\rm in} = - V_{\rm in}  (B_j^\dagger V_{\rm in}) + (\id_{SE} - V_{\rm in}V_{\rm in}^\dagger) B_j =  V_{\rm in}  (V_{\rm in}^\dagger B_j) + (\id_{SE} - V_{\rm in}V_{\rm in}^\dagger) B_j = B_j. 
\]
\eqref{eq:app-deriv} is then satisfied. 
Note that we can shift each $G_j$ by a constant to make it traceless as it doesn't affect the saturation of the original requirement \eqref{eq:app:orgreq}. Similarly, we can assume $G_j$ are linearly independent because otherwise the derivative of $\mE_\theta$ with respect to some parameter can be zero from \eqref{eq:app:orgreq}, violating the identifiability condition. 

In order to directly apply the SDP algorithm \eqref{eq:app-final}, we also need to perform reparameterization such that the orthonormality condition \eqref{eq:app-orthonormal} is satisfied. Now we have 
\[
\{G_j \}_{j=1}^m
\]
linearly independent and 
\[
\{\id_S \otimes H_a\}_{a=1}^{r^2 - 1} 
\]
orthonormal. Furthermore, $\{G_j \}_{j=1}^m \cup \{\id_S \otimes H_a\}_{a=1}^{r^2 - 1} $ are linearly independent because if any Hermitian operator in the $\{G_j\}$-span can be expressed as a linear combination of $\id_S \otimes H_a$, $\{\partial_j \mE_\theta\}$ will be linearly dependent (from \eqref{eq:app:orgreq}), violating the identifiability condition. Then we can find 
\[
\{ \widetilde{G_j} \}_{j=1}^m
\]
through reparameterization of $(\theta,\varphi)$ (equivalent to linear transformation on $G,H$) such that together with $\{\id_S \otimes H_a\}_{a=1}^{r^2 - 1}$ we have an orthonormal basis acting on $\mH_{SE}$ that spans the same subspace as $\{G_j \}_{j=1}^m \cup \{\id_S \otimes H_a\}_{a=1}^{r^2 - 1}$.  Specifically, viewing $G$ as a row vector of block matrices $(G_1 \cdots G_m)$ and similarly for $H$, assume
\[
\begin{pmatrix}
    \widetilde G & H
\end{pmatrix}
= \begin{pmatrix}
    G & H
\end{pmatrix}
\begin{pmatrix}
    T_1 & 0 \\
    T_2 & \id 
\end{pmatrix}. 
\]
The new parameters $(\tilde\theta, \tilde\varphi)$ are 
\[
\begin{pmatrix}
    \theta\\\varphi 
\end{pmatrix}
= 
\begin{pmatrix}
    T_1 & 0 \\
    T_2 & \id 
\end{pmatrix}
\begin{pmatrix}
    \tilde{\theta}\\\tilde{\varphi} 
\end{pmatrix}.
\]
Then 
\[
C_{\rm H}(\rho_{\theta,\varphi},W^*) = C_{\rm H}(\rho_{\tilde\theta,\tilde\varphi},\widetilde{W}^*),
\]
where 
\[
\widetilde{W}^* = 
\begin{pmatrix}
    \widetilde{W} & 0 \\
    0 & 0
\end{pmatrix} = \begin{pmatrix}
    T_1 & 0 \\
    T_2 & \id 
\end{pmatrix}^\top \begin{pmatrix}
    W & 0 \\
    0 & 0
\end{pmatrix}  \begin{pmatrix}
    T_1 & 0 \\
    T_2 & \id 
\end{pmatrix}
= \begin{pmatrix}
   T_1^\top W T_1 & 0 \\
    0 & 0
\end{pmatrix}.
\]
To conclude, \eqref{eq:app-final} applies to the case of general channels, after replacing $G_i \otimes \id_E$ in \eqref{eq:app-unitary-Q--} with $\widetilde{G}_i$ and $W$ with $\widetilde{W}$. The choice of $\{R_\nu\}$ also needs to be updated to keep $\{\widetilde{G}_i\}\cup\{\id_E \otimes H_a\}\cup\{R_\nu\}$ an orthonormal basis.

\end{document}